\documentclass[a4,12pt]{article}
\usepackage{graphicx,fancyhdr,mathrsfs,bbm}
\usepackage{amsfonts}
\usepackage{amsmath}
\usepackage{mathtools}
\usepackage{multicol}
\usepackage{mathrsfs}
\usepackage{multirow}
\usepackage{amssymb}
\usepackage{url}
\usepackage[colorlinks=true,citecolor=blue,urlcolor=blue]{hyperref}
\usepackage{fancyhdr}
\usepackage{indentfirst}
\usepackage{amsthm}
\usepackage{color}
\usepackage{comment}
\usepackage{dsfont}
\usepackage{natbib}
\usepackage[misc]{ifsym}
\usepackage[toc,page]{appendix}
\usepackage{subfig}
\usepackage[top=20mm, bottom=30mm, left=20mm, right=20mm]{geometry}
\usepackage{enumitem}
\setlist{
  align=left,
  labelindent=0mm,
  leftmargin=!,
  itemindent=0mm,
  listparindent=\parindent,
  parsep=0mm,
  topsep=1mm,
  itemsep=1mm
}
\setlist[itemize,1]{label={\mysquare}, labelwidth=\widthof{\mysquare\ }}
\setlist[itemize,2]{label={\mytriangle}, labelwidth=\widthof{\mytriangle\ }}
\setlist[itemize,3]{label={\mybar}, labelwidth=\widthof{\mybar\ }}
\setlist[itemize,4]{label={\mydot}, labelwidth=\widthof{\mydot\ }}
\setlist[enumerate,1]{label=(\arabic*), labelwidth=\widthof{(9)}}
\setlist[enumerate,2]{label=(\arabic{enumi}.\arabic*), labelwidth=\widthof{(9.9)}}
\setlist[enumerate,3]{label=(\arabic{enumi}.\arabic{enumii}.\arabic*), labelwidth=\widthof{(9.9.9)}}
\setlist[enumerate,4]{label=(\arabic{enumi}.\arabic{enumii}.\arabic{enumiii}.\arabic*), labelwidth=\widthof{(9.9.9.9)}}
\usepackage{actuarialsymbol}
\usepackage{actuarialangle}

\newcommand{\R}{\mathbb{R}}

\newcommand{\id}{\mathds{1}}

\renewcommand{\ge}{\geqslant}
\renewcommand{\le}{\leqslant}

\renewcommand{\epsilon}{\varepsilon}

\usepackage{array}
\newcommand{\PreserveBackslash}[1]{\let\temp=\\#1\let\\=\temp}
\newcolumntype{C}[1]{>{\PreserveBackslash\centering}p{#1}}
\newcolumntype{R}[1]{>{\PreserveBackslash\raggedleft}p{#1}}
\newcolumntype{L}[1]{>{\PreserveBackslash\raggedright}p{#1}}

\usepackage{csquotes}
\renewcommand{\mkbegdispquote}[2]{\itshape}

\usepackage{booktabs,array}

\newcount\rowc

\makeatletter
\def\ttabular{%
  \hbox\bgroup
  \let\\\cr
  \def\rulea{\ifnum\rowc=\@ne \hrule height 1.3pt \fi}
  \def\ruleb{
    \ifnum\rowc=1\hrule height 1.3pt \else
      \ifnum\rowc=6\hrule height \heavyrulewidth
      \else \hrule height \lightrulewidth\fi\fi}
  \valign\bgroup
  \global\rowc\@ne
  \rulea
  \hbox to 10em{\strut \hfill##\hfill}%
  \ruleb
  &&%
  \global\advance\rowc\@ne
  \hbox to 10em{\strut\hfill##\hfill}%
  \ruleb
  \cr}
\def\endttabular{%
  \crcr\egroup\egroup}

\newcommand{\rd}{\mathrm{d}}

\theoremstyle{definition}
\newtheorem{Theorem}{Theorem}

\newtheorem{Lemma}{Lemma}
\newtheorem{Proposition}{Proposition}
\newtheorem{Definition}{Definition}

\newtheorem{Remark}{Remark}

\usepackage{tikz}

\usepackage[compact]{titlesec}

\usepackage{color,xcolor,comment}
\usepackage{bm}

\renewcommand*{\P}{\mathbb{P}}

\usepackage{algpseudocode}
\usepackage{algorithm,algcompatible} 

\algnewcommand\algorithmicinput{\textbf{Input:}}
\algnewcommand\INPUT{\item[\algorithmicinput]}

\algnewcommand\algorithmicoutput{\textbf{Output:}}
\algnewcommand\OUTPUT{\item[\algorithmicoutput]}

\makeatletter
\DeclareRobustCommand{\bsquare}{%
  \mathop{\vphantom{\sum}\mathpalette\bigstar@\relax}\slimits@
}

\newcommand{\bigstar@}[2]{%
  \vcenter{%
    \sbox\z@{$#1\sum$}%
    \hbox{\resizebox{.9\dimexpr\ht\z@+\dp\z@}{!}{$\m@th\dsquare$}}%
  }%
}
\makeatother

\usepackage[onehalfspacing]{setspace}

\begin{document}

\title{
  Tail copula representation of path-based maximal tail dependence
}

\author{
  Takaaki Koike\thanks{\protect\linespread{1}\protect\selectfont Corresponding author.}\,\,\thanks{\protect\linespread{1}\protect\selectfont
    Graduate School of Economics, Hitotsubashi University.
    Email: \texttt{takaaki.koike@r.hit-u.ac.jp}},
  Marius Hofert\thanks{\protect\linespread{1}\protect\selectfont
    Department of Statistics and Actuarial Science, The University of Hong Kong.}\, and
  Haruki Tsunekawa\thanks{\protect\linespread{1}\protect\selectfont
    Graduate School of Economics, Hitotsubashi University.
    }
}

\maketitle

\begin{abstract}
  The classical tail dependence coefficient (TDC) may fail to capture
  non-exchangeable features of tail dependence due to its restrictive focus on
  the diagonal of the underlying copula. To address this limitation, the
  framework of path-based maximal tail dependence has been proposed, where a path of
  maximal dependence is derived to capture the most pronounced feature of
  dependence over all possible paths, and the path-based maximal TDC serves as a
  natural analogue of the classical TDC along this path. However, the
  theoretical foundations of path-based tail analyses, in particular the
  existence and analytical tractability, have remained limited. This paper
  addresses this issue in several ways. First, we prove the existence of a path
  of maximal dependence and the path-based maximal TDC when the underlying
  copula admits a non-degenerate tail copula.  Second, we obtain an explicit characterization of the maximal TDC in terms of the
  tail copula.
  Third, we show that  the first-order asymptotics of a path of maximal dependence is characterized by a one-dimensional optimization involving the tail
  copula.
   These results improve the analytical and
  computational tractability of path-based tail analyses.  As an application, we
  derive the asymptotic behavior of a path of maximal dependence for the
  bivariate $t$-copula and the survival Marshall--Olkin copula.
\end{abstract}


\section{Introduction}\label{sec:intro}
Copulas are a popular tool for modeling stochastic dependence, in particular
extremal dependence in the joint tails of the distribution of interest in
insurance and risk management applications; see~\citet{nelsen2006introduction},
\citet{jaworskidurantehaerdlerychlik2010} or
\citet[Chapter~7]{mcneil2015quantitative}. For notational convenience, we focus
on the lower tail in this work.
Results for other tails can be obtained by
suitable rotations or reflections, see
\citet[Section~3.4.1]{hofertkojadinovicmaechleryan2018}. Let $C$ be a copula and
$(U,V)\sim C$.  If it exists, the (lower) \emph{tail dependence coefficient
  (TDC)} of \citet{sibuya1960bivariate} is given by
\begin{align*}
  \lambda(C)=\lim_{u\downarrow 0}\frac{\mathbb{P}(U\le u,V\le u)}{u}=\lim_{u\downarrow 0}\frac{C(u,u)}{u}.
\end{align*}
The TDC is widely used to quantify the degree of dependence in the joint tail of
bivariate copulas, and, reminiscent of the notion of correlation, via matrices
of pairwise TDCs in the multivariate case \citep{embrechtshofertwang2016}.
However, since the definition of $\lambda(C)$ is based solely on the diagonal
$(u,u)_{u\in(0,1]}$ of the underlying copula $C$, the TDC is known to
potentially overlook off-diagonal features of tail dependence.

To capture off-diagonal tail dependence, various measures have been proposed
in the literature. In terms of copulas, the tail copula~\citep{schmidt2006non},
tail dependence function~\citep{joe2010tail} and tail order
function~\citep{hua2011tail} describe tail dependence in functional form.
Measures to numerically quantify the degree of off-diagonal tail dependence
have been proposed, for example, by~\citet{krupskii2015tail},
\citet{lee2018tail}, \citet{hua2019assessing} and \citet{siburg2024comparing}.

\citet{furman2015paths} proposed a path-based analysis to capture
off-diagonal bivariate tail dependence.
They introduced a \emph{path of maximal dependence}, denoted by
$(\varphi^\ast(u),\psi^\ast(u))_{u \in (0,1]}\subseteq[0,1]^2$.  For each $u$,
the point $(\varphi^\ast(u),\psi^\ast(u))$ maximizes the joint probability
$\mathbb{P}((U,V) \in [0,\varphi(u)]\times [0,\psi(u)])$ over all
\emph{admissible paths} $(\varphi(u),\psi(u))_{u \in (0,1]}\subseteq[0,1]^2$,
which satisfy $\lim_{u\downarrow 0}\varphi(u)=\lim_{u\downarrow 0}\psi(u)=0$, $\varphi(1)=\psi(1)=1$ and, for
every $u \in (0,1]$, that the rectangle $[0,\varphi(u)]\times [0,\psi(u)]$ has the
same area as the square $[0,u]^2$, namely, $u^2$; see
Definition~\ref{def:path:max:index} for a formal definition. The area condition
implies that $\psi(u)=u^2/\varphi(u)$.  Consequently, the path
$(\varphi^\ast(u),u^2/\varphi^\ast(u))_{u \in (0,1]}$ and associated indices
capture the most pronounced feature of dependence, which may be overlooked if
only the diagonal path $(u,u)_{u \in (0,1]}$ is considered.

The natural path-based extension of the TDC is thus the \emph{path-based maximal
  TDC} defined by
\begin{align*}
  \lambda_{\varphi^\ast}(C)=\lim_{u\downarrow 0}\frac{\mathbb{P}(U\le \varphi^\ast(u),V\le u^2/\varphi^\ast(u))}{u}=\lim_{u\downarrow 0}\frac{C(\varphi^\ast(u), u^2/\varphi^\ast(u))}{u}.
\end{align*}
Since~\citet{furman2015paths} define this coefficient under the assumption that
a path of maximal dependence
$(\varphi^\ast(u),u^2/\varphi^\ast(u))_{u \in (0,1]}$ exists, we denote it by
$\lambda_{\varphi^\ast}(C)$ instead of their original notation
$\lambda_{\operatorname{L}}^\ast$ to emphasize the underlying function $\varphi^\ast$.
Note, however, that
the quantity $\lambda_{\varphi^\ast}(C)$ is invariant under the choice of a path
of maximal dependence if such a path is not unique;
see~\citet[Section~2]{furman2015paths}.

Despite the natural intuition behind this coefficient, closed-form expressions for
$\lambda_{\varphi^\ast}(C)$ are rarely found in the literature, which can be
attributed to the difficulty of finding the function $\varphi^\ast$; see
\citet[Section~6]{furman2015paths}. Exceptions include the symmetric bivariate Gaussian
copula with positive correlation parameter and a subclass of bivariate Archimedean
copulas, for which one can show that the path of maximal dependence is the diagonal; see
\citet[Section~6.2]{furman2015paths} and \citet{furman2016tail}.
An example of non-exchangeable copulas with non-diagonal paths of maximal
dependence is the Marshall--Olkin (MO) copula family of
\citet{marshall1967generalized, marshall1967multivariate};
see~\citet[Section~4]{furman2015paths}.

To address these limitations in the framework of path-based maximal tail dependence, our main contribution is to reveal the close connection between
$\varphi^\ast$ and $\lambda_{\varphi^\ast}$ to the \emph{tail copula}~\citep{schmidt2006non}
\begin{align*}
  \Lambda(x,y;C)=\lim_{t\downarrow 0}\frac{C(tx,ty)}{t},\quad (x,y)\in (0,\infty)^2,
\end{align*}
and to the \emph{maximal tail concordance measure}~\citep[\emph{MTCM},][]{koike2023measuring}
\begin{align*}
  \lambda^\ast(C)=\sup_{b \in (0,\infty)}\Lambda\left(b,\frac{1}{b};C\right).
\end{align*}
If $b\mapsto \Lambda(b,1/b;C)$  has a unique maximizer in $(0,\infty)$, we denote it by $b^\ast$.
We establish the following results for any copula $C$ admitting a non-degenerate tail copula $\Lambda$:
\begin{enumerate}[label=(\roman*), labelwidth=\widthof{(iii)}]
\item\label{contr:i} a function of maximal dependence $\varphi^\ast$ and the maximal TDC $\lambda_{\varphi^\ast}(C)$ exist;
\item the two indices $\lambda_{\varphi^\ast}(C)$ and $\lambda^\ast(C)$ coincide; and
\item\label{contr:iii} if the unique maximizer $b^\ast$ exists, then $\varphi^\ast(u)$ is asymptotically equal to $b^\ast\,u$ as $u \downarrow 0$.
\end{enumerate}
These findings overcome key limitations encountered in the path-based analyses
of tail dependence by eliminating the need to verify the existence of
$\varphi^\ast$ and by reducing the analysis of $\varphi^\ast$ and its
first-order asymptotic behavior to a one-dimensional
optimization problem based on the tail copula $\Lambda$.

Note that closed-form expressions for the MTCM and $b^\ast$ are known for
various non-exchangeable copulas; see~\citet{koike2023measuring}
and~\citet{hofert2025w}.  As an application of our theoretical results, we
derive the asymptotic behavior of the path of maximal dependence for the bivariate
$t$-copula and the survival Marshall--Olkin copula.  For the bivariate $t$-copula, we show
that the diagonal is asymptotically the unique path of maximal dependence, and
hence, the maximal TDC coincides with the standard TDC.  Our proof is based on
the spectral representation of the tail copula of the bivariate $t$-copula,
which is derived from the corresponding $t$-\emph{extreme-value (EV) copula};
see~\citet{demarta2005t} and \citet{nikoloulopoulos2009extreme}.  For the survival Marshall--Olkin
copula, we show that any path of maximal dependence asymptotically coincides
with its singular curve.

The paper is organized as follows. After formally introducing the path-based
framework of maximal tail dependence and the MTCM in Section~\ref{sec:preliminaries},
we present the aforementioned main results in Section~\ref{sec:equivalence},
accompanied by analytical examples and
simulations. Section~\ref{sec:applications} contains the aforementioned application of our
theoretical findings to $t$-copulas and survival Marshall--Olkin copulas.
Section~\ref{sec:concl} provides a conclusion. All proofs are deferred to the
appendix.

\section{Preliminaries}\label{sec:preliminaries}
We start by introducing the concept of path-based maximal
dependence of~\citet{furman2015paths} for measuring off-diagonal tail
dependence.

\begin{Definition}[Path-based maximal tail dependence]\label{def:path:max:index}
  Let $C$ be a bivariate copula.
  \begin{enumerate}[label=(\arabic*)]
  \item A measurable function $\varphi:(0,1]\to [0,1]$ is called
    \emph{admissible} if
    \begin{enumerate}[label=(\roman*), labelwidth=\widthof{(iii)}]
    \item $\varphi(u) \in [u^2,1]$ for every $u\in(0,1]$; and
    \item $\lim_{u\downarrow 0}\varphi(u)=\lim_{u\downarrow 0}\left(u^2/\varphi(u)\right)=0$.
    \end{enumerate}
    Denote by $\mathcal A$ the set of all admissible functions.
  \item For $\varphi\in\mathcal A$, let
    \begin{align*}
      \Pi_\varphi(u) := C\left(\varphi(u), \frac{u^2}{\varphi(u)}\right),
      \qquad u\in(0,1].
    \end{align*}
    A function $\varphi^\ast\in \mathcal A$ is called a \emph{function of maximal dependence}
    if
    \begin{align*}
      \Pi_{\varphi^\ast}(u)=\max_{\varphi\in\mathcal{A}} \Pi_\varphi(u)
      \quad\text{for all $u\in(0,1]$}.
    \end{align*}
  \item If $C$ admits a function of maximal dependence $\varphi^\ast\in\mathcal A$,
    the \emph{path-based maximal tail dependence coefficient} is defined by
    \begin{align*}
      \lambda_{\varphi^\ast}(C) = \lim_{u\downarrow 0} \frac{\Pi_{\varphi^\ast}(u)}{u},
    \end{align*}
    provided the limit exists.
  \end{enumerate}
\end{Definition}

According to \citet[Theorem~2.3]{furman2015paths}, if $C$ has a unique function
of maximal dependence $\varphi^\ast$, then $\varphi^\ast$ is
continuous. Moreover, 
even if $C$ admits multiple functions of maximal dependence, the value of the measure
$\lambda_{\varphi^\ast}(C)$ does not depend on the specific choice of
$\varphi^\ast$.

For an admissible function $\varphi\in\mathcal A$, the path
$\left(\varphi(u), u^2/\varphi(u)\right)_{u\in(0,1]}$ approaches $(0,0)$ while
the area of the rectangle $[0,\varphi(u)]\times[0,u^2/\varphi(u)]$ is $u^2$,
independently of $\varphi$.  By definition, for a function of maximal dependence
$\varphi^\ast$, the $C$-volume of $[0,\varphi(u)]\times[0,u^2/\varphi(u)]$,
equivalently $\P(U\le\varphi(u), V\le u^2/\varphi(u))$, is maximal
among all admissible choices.  Clearly, if $\varphi^\ast$ is the identity, then
$\lambda_{\varphi^\ast}$ coincides with the TDC $\lambda$.

Next, we introduce a measure of off-diagonal tail dependence introduced
by~\citet{koike2023measuring}. To this end, the \emph{tail copula}
$\Lambda: (0,\infty)^2\rightarrow [0,\infty)$ of a bivariate copula $C$ is
defined by
\begin{align*}
  \Lambda(x,y;C)=\lim_{t \downarrow 0}\frac{C(tx,ty)}{t},\quad (x,y)\in (0,\infty)^2,
\end{align*}
provided the limit exists; see~\citet{schmidt2006non} for basic properties. Note that $\Lambda(1,1)$ corresponds to the TDC.
If $\Lambda$ is not identically $0$, then we say that $\Lambda$
is \emph{non-degenerate}, otherwise \emph{degenerate}.

\begin{Definition}[Maximal tail concordance measure]
  Let $C$ be a bivariate copula admitting a non-degenerate tail copula $\Lambda$.
  We call the
  function $b\mapsto \Lambda(b,1/b;C)$ on $(0,\infty)$ the \emph{profile tail copula}.
  The \emph{maximal tail concordance measure (MTCM)}
  is then defined by
  \begin{align}\label{eq:def:mtcm}
    \lambda^\ast(C) = \sup_{b\in (0,\infty)} \Lambda\left(b,\frac{1}{b};C\right).
  \end{align}
  The unique maximizer $b\in(0,\infty)$ of the profile tail
  copula, if it exists, is denoted by $b^{\ast}=b^{\ast}(C)$.
\end{Definition}

The MTCM quantifies the maximal possible tail probability
\begin{align*}
  \lim_{t \downarrow 0}\frac{C(tb,t/b)}{t}=\lim_{t \downarrow 0}\frac{\mathbb{P}((U,V)/t \in[0,b]\times[0,1/b])}{t}
\end{align*}
over all possible rectangles $[0,b]\times[0,1/b]$, $b\in(0,\infty)$, with unit
area. Basic properties of the MTCM can be found
in~\citet[Proposition~3.7~(1)]{koike2023measuring}.  In particular, the
supremum in~\eqref{eq:def:mtcm} is always attainable and we can thus write
\begin{align*}
  \lambda^\ast(C)=\max_{b \in (0,\infty)}\Lambda\left(b,\frac{1}{b};C\right);
\end{align*}
see~\citet[Remark~3.10]{koike2023measuring}.

\section{Equivalence between the two measures}\label{sec:equivalence}
This section provides the main contributions~\ref{contr:i}--\ref{contr:iii} stated in Section~\ref{sec:intro},
as well as a numerical illustration.

\subsection{Equivalence result}
Assuming the existence of the non-degenerate tail copula, the following theorem
implies that a copula admits $\varphi^\ast$ and $\lambda_{\varphi^\ast}$, and
that the one-dimensional optimization problem underlying the MTCM completely
determines the asymptotic behavior of $\varphi^\ast$ and thus of
$\lambda_{\varphi^\ast}(C)$.

\begin{Theorem}[Equivalence between $\lambda_{\varphi^\ast}$ and $\lambda^\ast$]\label{thm:equivalence}
  Let $C$ be a bivariate copula admitting a non-degenerate tail copula $\Lambda$.
  Then the following statements hold.
  \begin{enumerate}[label=(\roman*), labelwidth=\widthof{(iii)}]
  \item\label{item:varphi:ast:existence} The copula $C$ admits a function of maximal dependence $\varphi^\ast\in\mathcal{A}$.
  \item The path-based maximal TDC $\lambda_{\varphi^\ast}(C)$ exists and satisfies $\lambda_{\varphi^\ast}(C)=\lambda^\ast(C)$.
  \item\label{item:varphi:ast:asymptotic} If, additionally, the supremum of the
    profile tail copula is uniquely attained at $b^\ast\in(0,\infty)$, then any
    function of maximal dependence $\varphi^\ast$ satisfies
  \begin{align*}
    \lim_{u\downarrow 0}\frac{\varphi^\ast(u)}{u}=b^\ast.
  \end{align*}
  \end{enumerate}
\end{Theorem}

Theorem~\ref{thm:equivalence} has various practical implications.  First, it
guarantees the existence of $\varphi^\ast$ and $\lambda_{\varphi^\ast}$ whenever
the tail copula $\Lambda$ of $C$ is non-degenerate. Second, it simplifies
path-based tail dependence analyses by allowing one to verify the existence of
the MTCM and the uniqueness of $b^\ast$, which is typically more straightforward
than directly deriving the function of maximal dependence. Consequently, if tail
behavior is of primary interest, it suffices to focus on finding $b^\ast$ since
it completely determines the asymptotic behavior of $\varphi^\ast$ and thus
$\lambda_{\varphi^\ast}$.

We discuss the possibility of extending Theorem~\ref{thm:equivalence} in the following remarks.

\begin{Remark}[Non-existence of $\varphi^\ast$]
  Not every copula possesses a function of maximal dependence
  $\varphi^\ast\in \mathcal A$.  To see this, consider the
  \emph{Farlie--Gumbel--Morgenstern (FGM) copula} with parameter $\theta = -1$,
  given by $C(u,v) = uv(1 - (1-u)(1-v))$, $(u,v)\in[0,1]^2$. Its tail copula is
  degenerate, so Theorem~\ref{thm:equivalence} does not apply.  It is
  straightforward to check that, for each $u\in (0,1)$, the maximum of
  $x \to C(x,u^2/x)$ on $[u^2,1]$ is attained at the two points $x=u^2$ and
  $x=1$.  Therefore, if $C$ admits a function of maximal dependence
  $\varphi^\ast\in\mathcal A$, then it has to satisfy
  $\varphi^\ast(u)\in\{u^2,1\}$ for each $u \in (0,1)$.  However, since
  $(\varphi^\ast(u), u^2/\varphi^\ast(u)) \in \left\{ (u^2, 1), (1, u^2)
  \right\}$ for every $u \in (0,1)$, the pair
  $(\varphi^\ast(u), u^2/\varphi^\ast(u))$ cannot converge to $(0,0)$.
  Therefore, $\varphi^\ast$ is not admissible.
\end{Remark}

\begin{Remark}[Uniqueness assumption of $b^\ast$]
  Let
  $\mathcal B^\ast =\mathcal B^\ast(C):= \operatorname{argmax}_{b \in
    (0,\infty)} \Lambda(b,1/b;C)$.
  Theorem~\ref{thm:equivalence}~\ref{item:varphi:ast:asymptotic} indicates that,
  if $|\mathcal B^\ast|=1$, then any function of maximal dependence
  $\varphi^\ast$ admits the same right-sided derivative at $0$.  We do not
  expect such a statement to hold when $\mathcal B^\ast$ is not a singleton.
  Indeed, suppose that $\mathcal B^\ast(C)=\{b_1,b_2\}$ for $b_1,b_2\in (0,\infty)$ with $b_1\neq b_2$ and that $C$ admits two functions of maximal dependence $\varphi_1^\ast$ and $\varphi_2^\ast$ such that $\lim_{u\downarrow 0}\varphi_1^\ast(u)/u=b_1$ and $\lim_{u\downarrow 0}\varphi_2^\ast(u)/u=b_2$.
  Then, for any measurable $A\subseteq (0,1]$, the function $\varphi_A=\varphi_1^\ast\,\id_A +\varphi_2^\ast\,\id_{A^c}$ is also
  a function of maximal dependence for $C$.
  If we take $A=\cup_{n\ge 1}(2^{-2n},2^{-2n+1}]$, the resulting function $\varphi_A\in \mathcal A$ with $\varphi_A(0)=0$ does not admit a right-sided
  derivative at $0$.
\end{Remark}

\subsection{Numerical illustration}
We now present numerical experiments in support of
Theorem~\ref{thm:equivalence}. We consider two non-exchangeable survival copulas
known to admit non-degenerate tail copulas.  Note that, for a copula $C$ and
$(U,V)\sim C$, its survival copula $\hat C$ is the copula of $(1-U,1-V)$ given
by $\hat C(u,v)=-1+u+v+C(1-u,1-v)$, $(u,v)\in[0,1]^2$.  Our first model is the
survival Marshall--Olkin copula $\hat C_{\alpha,\beta}^\text{MO}$ (with
parameters $\alpha=0.35$ and $\beta=2\alpha=0.7$), where the
\emph{Marshall--Olkin copula} is defined by
$C_{\alpha,\beta}^\text{MO}(u,v)=\min(u^{1-\alpha}v, uv^{1-\beta})$,
$(u,v)\in[0,1]^2$. Our second example is the survival asymmetric Gumbel copula
$\hat C^{\text{AG}}_{\alpha,\beta,\theta}$ (with parameters $\alpha=0.35$,
$\beta=2 \alpha=0.7$ and $\theta=2$), where the \emph{asymmetric Gumbel copula}
is defined by
$C^{\text{AG}}_{\alpha,\beta,\theta}(u,v)=e^{\ln(uv)A(\ln(v)/\ln(uv))}$,
$(u,v)\in[0,1]^2$, for the parametric Pickands dependence
function 
\begin{align*}
  A_{\alpha,\beta,\theta} (w)=(1-\alpha)w + (1-\beta)(1-w) + \{(\alpha w)^\theta + (\beta(1-w))^\theta\}^{1/\theta}, \quad w\in [0,1],
\end{align*}
with $\alpha,\beta \in (0,1]$ and $\theta>1$;
see~\citet[Section~4.15]{joe2015dependence}. The tail copula of
$\hat C^{\text{AG}}_{\alpha,\beta,\theta}$ and its MTCM (for general parameters)
can be found in~\citet{koike2023measuring}. In particular, the MTCM is uniquely
attained at $b^\ast=\sqrt{\beta/\alpha}$.

The two rows in Figure~\ref{fig} show the results for the survival Marshall--Olkin
copula $\hat C_{\alpha,\beta}^\text{MO}$ and the survival asymmetric Gumbel copula
$\hat C^{\text{AG}}_{\alpha,\beta,\theta}$, respectively.
\begin{figure}[htbp]
  \centering
  \includegraphics[height=0.24\textheight]{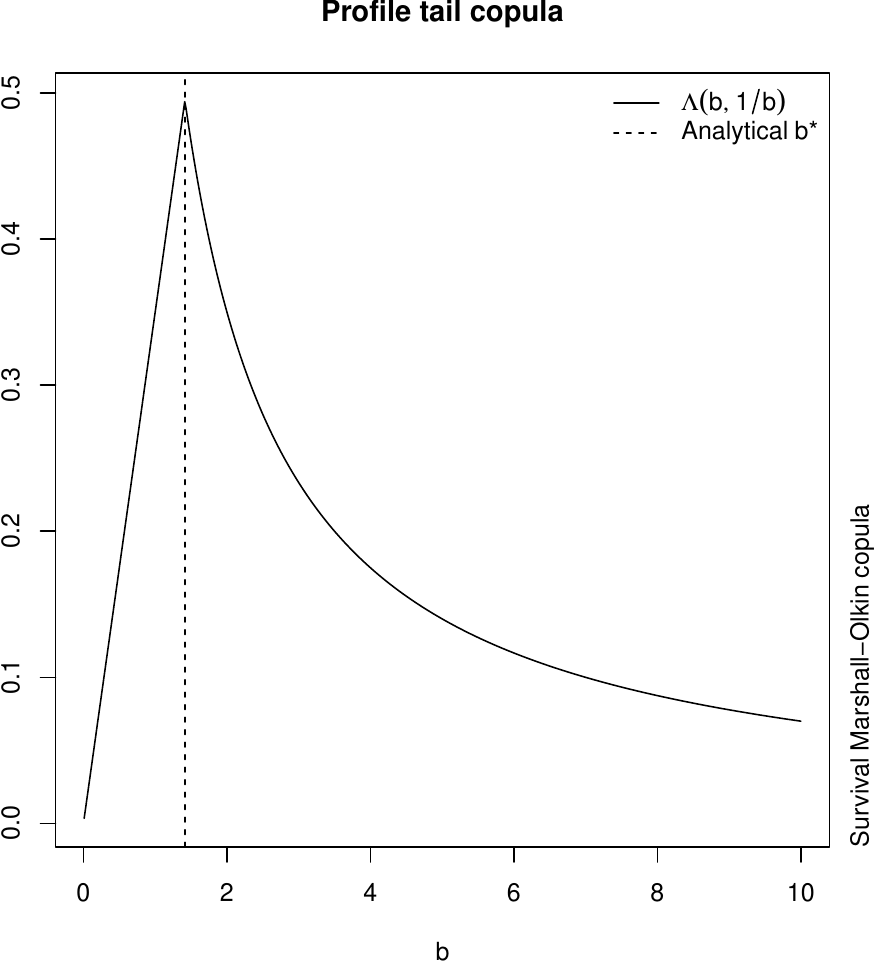}\hfill
  \includegraphics[height=0.24\textheight]{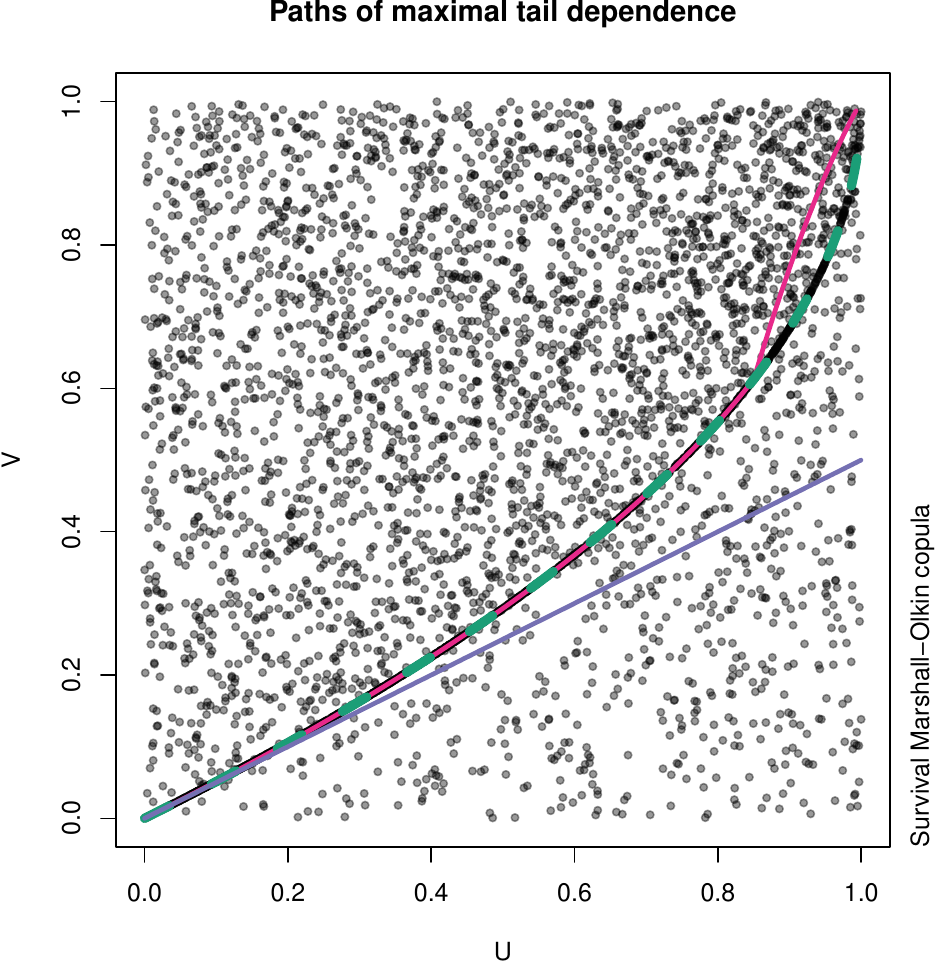}\hfill
  \includegraphics[height=0.24\textheight]{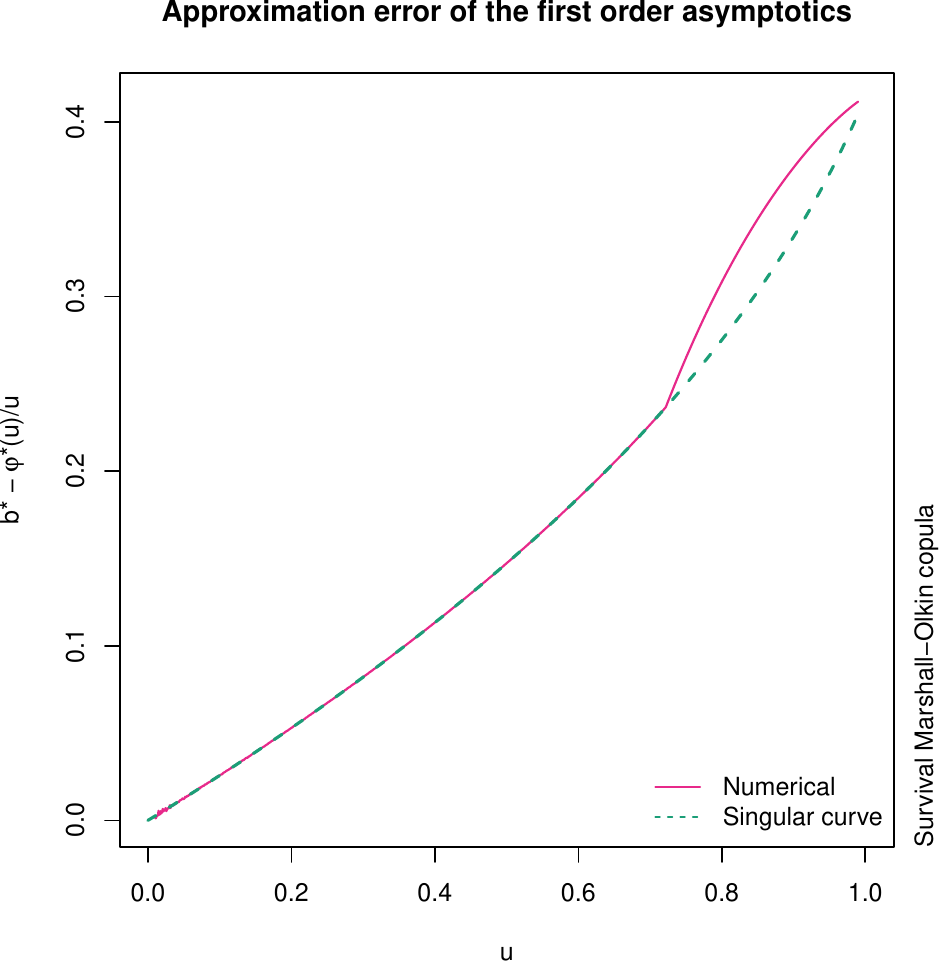}\\[3mm]
  \includegraphics[height=0.24\textheight]{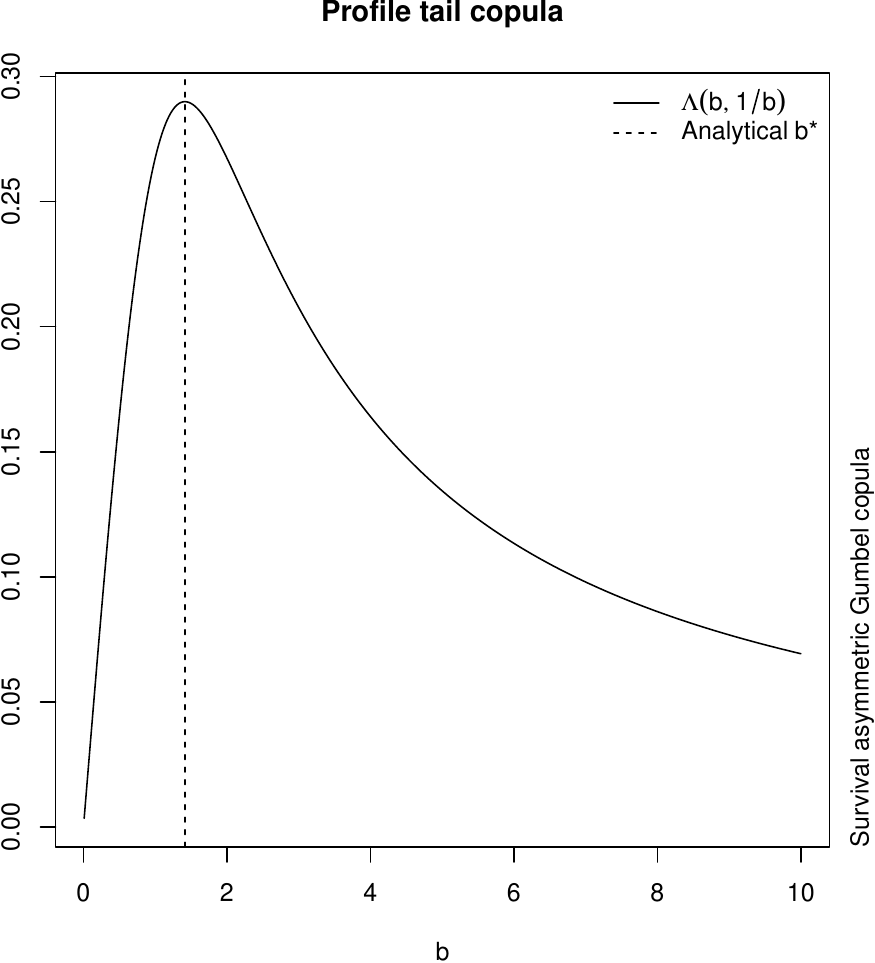}\hfill
  \includegraphics[height=0.24\textheight]{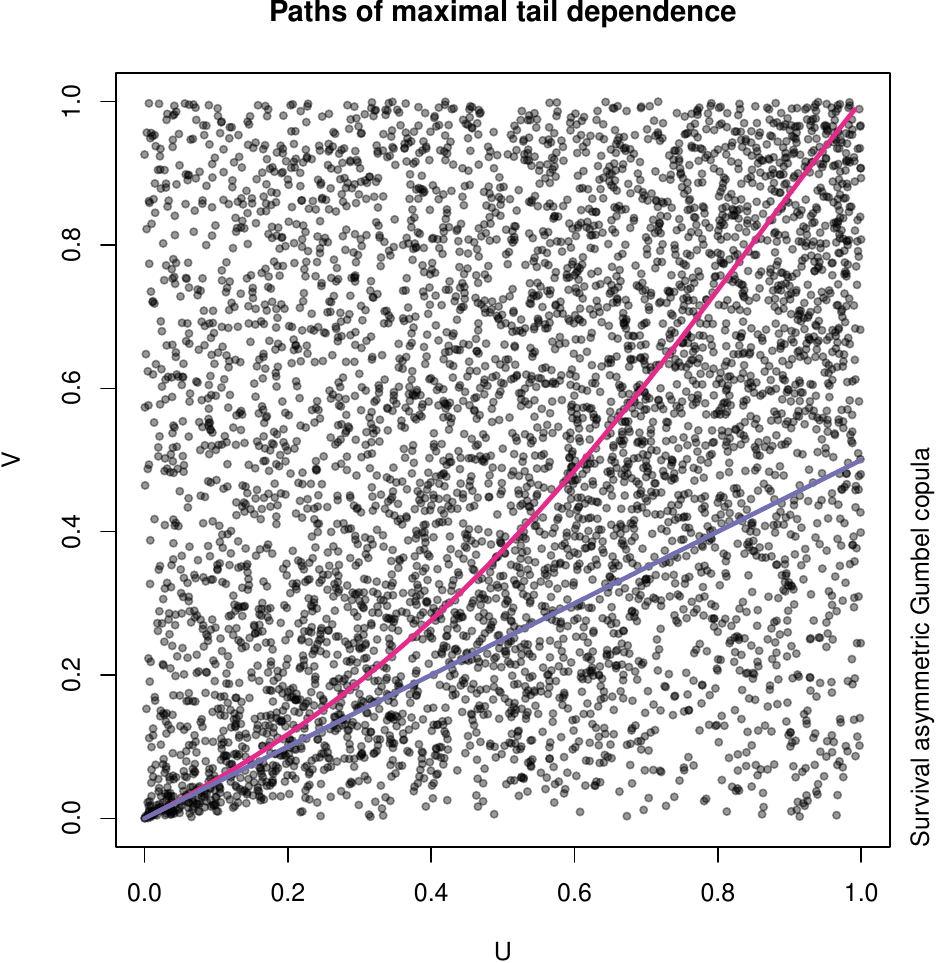}\hfill
  \includegraphics[height=0.24\textheight]{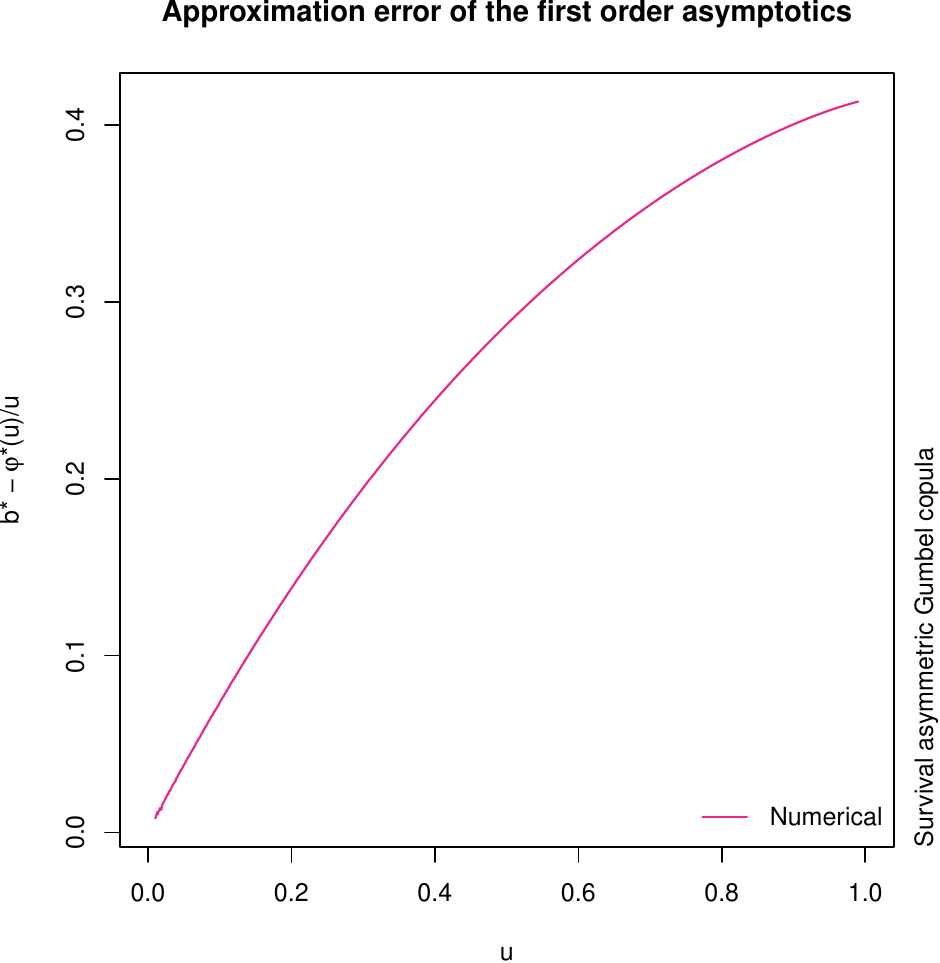}
  \caption{Comparison of the path-based maximal TDC $\lambda_{\varphi^\ast}(C)$
    and the MTCM $\lambda^\ast(C)$ for a survival Marshall--Olkin copula
    $\hat C_{\alpha=0.35,\beta=0.7}^\text{MO}$ (first row) and a survival
    asymmetric Gumbel copula
    $\hat C^{\text{AG}}_{\alpha=0.35,\beta=0.7,\theta=2}$ (second row).  The
    first column displays profile tail copula plots $b\mapsto \Lambda(b,1/b)$,
    with the vertical dashed lines indicating $b^\ast$, derived analytically.
    The second column shows scatter plots of size $5000$ of the respective
    copulas, overlaid with the numerically searched path of maximal dependence
    $(\varphi^\ast(u),u^2/\varphi^\ast(u))$ (red), the straight line passing
    through $(b^\ast,1/b^\ast)$ (blue), and, only available for
    $\hat C_{\alpha=0.35,\beta=0.7}^\text{MO}$, the singular component
    $(x_u^\ast,u^2/x_u^\ast)_{u\in (0,1]}$ (green, dashed,
    see~\eqref{eq:curve:smo} for $x_u^\ast$).  The third column displays the
    difference $u\mapsto b^\ast - \varphi^\ast(u)/u$ (red).  The green dashed
    line, only available for $\hat C_{\alpha=0.35,\beta=0.7}^\text{MO}$,
    displays $u\mapsto b^\ast - x_u^\ast/u$.}
  \label{fig}
\end{figure}
The first column shows their profile tail copula $b\mapsto\Lambda(b,1/b;C)$,
where the vertical dashed line indicates the unique maximizer $b^\ast$, computed
analytically. The second column displays a scatter plot of size $5000$ from the
respective model, overlaid with the numerically searched path of maximal dependence
$(\varphi^\ast(u),u^2/\varphi^\ast(u))_{u\in(0,1]}$ (red) and the straight line passing
through $(0,0)$ and $(b^\ast,1/b^\ast)$ (blue).  For
$\hat C_{\alpha,\beta}^\text{MO}$, the singular component (green, dashed) is
also included, which is found analytically. In Remark~\ref{rem:closed:form:beta:2alpha} in Section~\ref{sec:SMO}, we will
specify the closed-form expression of this curve as $(x_u^\ast,u^2/x_u^\ast)$,
where
\begin{align}\label{eq:curve:smo}
  x_u^\ast = 2u\sqrt{\frac{2}{3}}\cos\left(\frac{1}{3}\arccos\left(-\frac{3\sqrt{6}u}{8}\right)\right).
\end{align}
We will also show in Proposition~\ref{prop:smo} in Section~\ref{sec:SMO} that
$\varphi^\ast(u)$, $x_u^\ast$ and $\sqrt{2}u$ are all asymptotically equivalent
as $u\downarrow 0$ for any function of maximal dependence $\varphi^\ast$ of
$\hat C_{\alpha,\beta}^\text{MO}$. A feature visible from Figure~\ref{fig} is
that $x_u^\ast$ deviates from the numerically searched $\varphi^\ast$ for $u$
close to $1$.  In other words, the singular curve does not in general coincide with the path of maximal dependence on the entire interval $(0,1)$.
The final column shows the difference $u\mapsto b^\ast - \varphi^\ast(u)/u$ to
assess the convergence of $\varphi^\ast(u)/u$ to $b^\ast$ for $u\downarrow 0$ as
stated in Theorem~\ref{thm:equivalence}~\ref{item:varphi:ast:asymptotic}. The
red solid line corresponds to the numerically searched path of maximal
dependence.  For the survival Marshall--Olkin copula, we also show
$u\mapsto b^\ast - x_u^{\ast}/u$ as green dashed line.

\section{Applications to existing copulas}\label{sec:applications}

\subsection{Bivariate $t$-copulas}

In this section we focus on bivariate $t$-copulas, which are known to be in the
domain of attraction of the corresponding $t$-EV
copulas~\citep{demarta2005t,nikoloulopoulos2009extreme}.  By deriving the
spectral representation of the $t$-EV copula, we will show that the profile tail
copula of the $t$-copula is uniquely maximized at $b^\ast=1$.

In the bivariate setting, a copula $C$ admits a tail copula if and only if $C$ is in the domain of attraction of an \emph{extreme-value (EV) copula} $C_\ell$, that is $\lim_{n \to \infty}C(u^{1/n},v^{1/n})^n = C_{\ell}(u,v)$, $(u,v)\in[0,1]^2$; see~\citet[][Section~2]{einmahl2008method} and \citet{gudendorf2010extreme}.
Here, $C_{\ell}$ is an EV copula defined by $C_\ell(u,v)=\exp\{
  -\ell(-\ln u,-\ln v)
  \}$, $(u,v) \in [0,1]^2$, for a \emph{stable tail dependence
  function} $\ell:[0,\infty)^2\to [0,\infty)$ satisfying convexity, $1$-homogeneity and
$\max(x,y)\le\ell(x,y)\le x+y$, $(x,y)\in[0,\infty)^2$ \citep{ressel2013}.
Note that $C_\ell$ is max-stable in the sense that $C_{\ell}(u^{1/n},v^{1/n})^n = C_{\ell}(u,v)$ for all $(u,v)\in[0,1]^2$ and all integers $n\ge 1$;~see~\citet{gudendorf2010extreme}.
Therefore, $C_\ell$ itself is in the domain of attraction of $C_\ell$.
The relationship between the tail copula $\Lambda$ of $C$ and the stable tail dependence function $\ell$ of $C$ is
 \begin{align*}
     \Lambda(x, y;\hat{C}) = x + y - \ell(x, y),\quad (x,y)\in(0,\infty)^2;
  \end{align*}
  see~\citet[][Section~2]{einmahl2008method}. Therefore, the tail copula
  $\Lambda(\cdot;\hat C)$ is non-degenerate if and only if
  $\ell(x,y)\not\equiv x+y$, that is if $C_{\ell}\neq\Pi$.

Recall that a stable tail dependence function $\ell$ is associated with the \emph{spectral (angular) measure} $H$ on $[0,1]$ via
\begin{align*}
  \ell(x,y)=\int_{[0,1]} \max\{wx,(1-w)y\}\,H({\rm d}w),\quad (x,y) \in [0,\infty)^2,
\end{align*}
where $H$ satisfies the constraints $\int_{[0,1]} w\,H({\rm d}w)=1$ and
$ \int_{[0,1]} (1-w)\,H({\rm d}w)=1$; see, for
example,~\citet{gudendorf2010extreme} for details.  Using this measure, the tail
copula can also be represented by
\begin{align}\label{eq:Lambda:H}
  \Lambda(x,y;\hat C)=\int_{[0,1]} \min\{wx,(1-w)y\}\,H({\rm d}w),
  \quad (x,y) \in (0,\infty)^2;
\end{align}
see~\citet[][Section~2]{einmahl2008method}.  Note that the tail copula is
degenerate if and only if $H=\delta_0+\delta_1$ for Dirac measures $\delta_0,\delta_1$, that is $H$ is entirely concentrated on $\{0,1\}$.  Using
this spectral representation, we analyze the MTCM and its attainer $b^\ast$ in
the following lemma.
\begin{Lemma}[MTCM and the spectral measure]\label{lem:mtcm:and:h}
 Let $C$ be a bivariate copula in the domain of attraction of an extreme-value copula $C_\ell$, where $\ell$ is a stable tail dependence function.
Assume that the associated spectral measure $H$ has a Lebesgue density $h$ on $(0,1)$.
   Then the following statements hold.
   \begin{enumerate}[label=(\roman*), labelwidth=\widthof{(iii)}]
\item For $m:\mathbb{R}\to \mathbb{R}$ given by
\begin{align}\label{eq:m:function}
  m(a)=2\{w(a)(1-w(a))\}^{3/2}\,h(w(a)),\quad w(a)=\frac{e^{2a}}{1+e^{2a}},
\end{align}
it holds that
\begin{align*}
    \lambda^\ast(\hat C)=\max_{s \in \R} \int_{-\infty}^{\infty} e^{-|s+a|}\,m(a)\,{\rm d}a.
\end{align*}
\item Suppose that $m$ in~\eqref{eq:m:function} is even, integrable on $\R$ and strictly decreasing on $(0,\infty)$.
Then the MTCM $\lambda^\ast(\hat C)$ is uniquely attained at $b^\ast=1$.
   \end{enumerate}
\end{Lemma}

By specifying the density $h$ of the spectral measure of a $t$-EV copula, it follows from Lemma~\ref{lem:mtcm:and:h} that the maximizer is $b^\ast=1$ for $t$-copulas.

\begin{Proposition}[Attaining $b^\ast$ for $t$-copulas]\label{prop:bivariate:t:b:1}
  Let $C_{\nu,\rho}$ 
  be the bivariate $t$-copula with degrees of freedom parameter $\nu>0$ and
  correlation parameter $\rho\in (-1,1)$.  Then its MTCM is uniquely attained at
  $b^\ast=1$.
\end{Proposition}

This proposition, together with Theorem~\ref{thm:equivalence}, states that the
diagonal is asymptotically the unique path of maximal dependence for
$C_{\nu,\rho}$, and hence its path-based maximal TDC
$\lambda_{\varphi^\ast}$ coincides with the standard TDC $\lambda$.

\subsection{Survival Marshall--Olkin copulas}\label{sec:SMO}

Next, we consider the survival Marshall--Olkin copula $\hat{C}_{\alpha,\beta}^\text{MO}$ with
parameters $(\alpha,\beta)\in(0,1]^2$. The survival Marshall--Olkin copula admits the following tail copula and MTCM. 

\begin{Lemma}[MTCM for $\hat{C}_{\alpha,\beta}^\text{MO}$]\label{lem:tail:copula:smo}
  The survival Marshall--Olkin copula $\hat C_{\alpha,\beta}^\text{MO}$ admits the tail
  copula
  \begin{align*}
    \Lambda(x,y;\hat C_{\alpha,\beta}^\text{MO})=\min(\alpha x,\, \beta y),\quad (x,y)\in(0,\infty)^2,
  \end{align*}
  with corresponding MTCM given by $\lambda^\ast(\hat C_{\alpha,\beta}^\text{MO})=\sqrt{\alpha\beta}$,
  where the maximum is uniquely attained at $b^\ast = \sqrt{\beta/\alpha}$.
\end{Lemma}

The next proposition shows that any function of maximal dependence for
$\hat C_{\alpha,\beta}^\text{MO}$ asymptotically coincides with its singular
curve. To this end, for two general functions $f,g:\mathbb{R}\to \mathbb{R}$, we
write $f(u)\simeq g(u)$, $u \to u^\ast \in [-\infty,\infty]$, if
$\lim_{u \to u^\ast}(f(u)/g(u))=1$.

\begin{Proposition}[Maximal tail dependence for $\hat{C}_{\alpha,\beta}^{\text{MO}}$]\label{prop:smo}
  Consider the survival Marshall--Olkin copula $\hat{C}_{\alpha,\beta}^{\text{MO}}$ with parameters $\alpha, \beta \in (0,1]$.
 \begin{enumerate}[label=(\roman*), labelwidth=\widthof{(iii)}]
\item\label{item:smo:singular:curve} For each $u \in (0,1]$, the equation (in $x$)
  \begin{align}
    (1-x)^\alpha = \left(1-\frac{u^2}{x}\right)^\beta\label{eq:singular:curve}
  \end{align}
  has a unique solution in $[u^2,1]$, denoted by $x_u^\ast$.
\item Any function of maximal dependence $\varphi^\ast$ satisfies the asymptotic
  equivalence $$\varphi^\ast(u) \simeq x_u^\ast \simeq u\sqrt{\frac{\beta}{\alpha}}\qquad\text{as}\quad
  u\downarrow 0.$$
  \end{enumerate}
\end{Proposition}

\begin{Remark}[Closed-form expression for $x_u^\ast$]\label{rem:closed:form:beta:2alpha}
  If $\beta=2\alpha$, $\alpha \in (0,1/2]$, Equation~\eqref{eq:singular:curve}
  reduces to the standard cubic equation $P_u(x)=0$, where
  $P_u(x)=x^3-2u^2x+u^4$.
  The discriminant
  \begin{align*}
    \Delta = \left(\frac{u^4}{2}\right)^2 + \left(\frac{-2u^2}{3}\right)^3
    = \frac{u^8}{4} - \frac{8u^6}{27}
    = u^6\left(\frac{u^2}{4} - \frac{8}{27}\right)
  \end{align*}
  is strictly negative for every $u \in (0,1]$, and thus there are three distinct real roots.
  In this case, the solutions are given by the trigonometric form of Cardano's formula.
  The three real roots, $x_u^{\ast\,(k)}$ for $k \in \{0,1,2\}$, without the restriction
  $x\in[u^2,1]$, are
  \begin{align}\label{eq:curve:k}
    x_u^{\ast\,(k)}
    =
    2u\sqrt{\frac{2}{3}}
    \cos\left(
      \frac{1}{3}\arccos\left(-\frac{3\sqrt{6}u}{8}\right)+\frac{2\pi k}{3}
    \right),
    \qquad k\in\{0,1,2\};
  \end{align}
  see~\citet[p.~124]{Turnbull1947}. By inspection of~\eqref{eq:curve:k}, one finds that
  $x_u^{\ast\,(1)}<0$ and $0<x_u^{\ast\,(2)}<x_u^{\ast\,(0)}\le 1$ for every $u\in(0,1]$.
  Since
  $ P_u(u^2)=u^4(u^2-1)\le0
    $ and $
    P_u(1)=(1-u^2)^2\ge0,
  $
  the polynomial $P_u$ has at least one root in $[u^2,1]$.
  Proposition~\ref{prop:smo}~\ref{item:smo:singular:curve} shows that the solution
  of~\eqref{eq:singular:curve} in $[u^2,1]$ is unique. Therefore the unique root in
  $[u^2,1]$ must be the larger of the two positive roots, namely $x_u^{\ast\,(0)}$.
  Note that, for $u$ close to $0$, we have
  $
    x_u^{\ast\,(0)} \simeq \sqrt{2}u$,
$    x_u^{\ast\,(1)} \simeq -\sqrt{2}u$
and $    x_u^{\ast\,(2)} \simeq u^2/2$
  by using the standard Taylor expansions $\arccos(z)=(\pi/2)-z+O(z^3)$, $\cos(z)=1-(z^2/2)+O(z^4)$ and $\sin(z)=z+O(z^3)$
  around $z=0$.
\end{Remark}

\section{Conclusion}\label{sec:concl}
We established a fundamental theoretical connection between two frameworks for
measuring off-diagonal tail dependence, namely, the path-based analysis of
tail dependence and the MTCM based on tail copulas. Through the lens of tail
copulas, we first proved the existence of a path of maximal dependence and the
corresponding path-based maximal TDC. Second, we established the equivalence
between the maximal TDC and the MTCM assuming the existence of a non-degenerate
tail copula.  Third, we derived the asymptotic behavior of a path of maximal
dependence near the origin. For the purpose of quantifying maximal tail
dependence along a path, our results imply that it is sufficient to study the
MTCM and its attainer, which are known to exhibit considerably improved
analytical and numerical tractability.  To demonstrate this tractability, we
studied the asymptotic behavior of a path of maximal dependence for $t$-copulas
and survival Marshall--Olkin copulas.  We showed that any path of maximal
dependence around the origin is diagonal for $t$-copulas and follows the
singular curve for survival Marshall--Olkin copulas.

\section*{Acknowledgements}
Takaaki Koike is supported by the Japan Society for the Promotion of Science (JSPS) KAKENHI grant numbers JP24K00273 and JP26K21178.




\newpage

\appendix

\section*{Appendix}

\section{Proofs}

This appendix collects all the proofs omitted in the main text.

\subsection{Theorem~\ref{thm:equivalence}}

\begin{proof}[Proof of Theorem~\ref{thm:equivalence}]
\mbox{}
\begin{enumerate}[label=(\roman*), labelwidth=\widthof{(iii)}]
\item
Throughout the proof
of~\ref{item:varphi:ast:existence}, all references to sections, theorems,
definitions, and lemmas are to \citet{aliprantis2006infinite}.  Our main tools
are Berge's maximum theorem and the Kuratowski-Ryll-Nardzewski selection theorem
in set-valued analysis; see Theorems 17.31 and 18.13, respectively.  We use the
term \emph{correspondence} to denote a set-valued map, and the notation
$\Phi: A \rightrightarrows B$ indicates that for every $u \in A$, the image
$\Phi(u)$ is a subset of $B$.  For a subset $E \subseteq B$, we denote by
$\Phi^{\operatorname{u}}(E)$ and $\Phi^{\ell}(E)$ the \emph{upper} and
\emph{lower} inverse images of $E$ under $\Phi$, respectively (Section 17.1).  A
function $\varphi: A \to B$ is called a \emph{selector} of $\Phi$ if
$\varphi(u) \in \Phi(u)$ for all $u \in A$.

Let $\Phi: A \rightrightarrows B$ be a correspondence, where $A$ and $B$ are
topological spaces with $A$ equipped with the Borel $\sigma$-algebra $\Sigma$
generated by its topology. 
Then Definition 17.2 defines which $\Phi$ are \emph{upper
  hemicontinuous (uhc)}, \emph{lower hemicontinuous (lhc)} and
\emph{continuous}.  There are also two concepts of measurability for $\Phi$:
\emph{weak measurability} and \emph{measurability}; see Definition 18.1.

Let $A=(0,1]$, $B=[0,1]$ and $\Gamma:A\rightrightarrows B$ be the correspondence $\Gamma(u)=[u^2,1]$.
Note that $\Gamma$ is continuous (in view of Theorem 17.15) and has non-empty compact values.
 Define
\begin{align*}
g:\operatorname{Gr}(\Gamma)\to\mathbb R,\qquad
g(u,x):=C\left(x,\frac{u^2}{x}\right),
\end{align*}
where $\operatorname{Gr}(\Gamma):=\{(u,x)\in A\times B:x\in\Gamma(u)\}$.
The map $g$ is continuous on the graph $\operatorname{Gr}(\Gamma)$. Let $m(u):=\max_{x\in\Gamma(u)} g(u,x)$ and define the maximizer correspondence $\Phi:A\rightrightarrows B$ by
\begin{align*}
\Phi(u):=\{x\in\Gamma(u): g(u,x)=m(u)\},\qquad u\in A.
\end{align*}
Since $B$ is Hausdorff, we have from Berge's maximum theorem (Theorem 17.31) that the value function $m$ is continuous on $A$; $\Phi$ has non-empty compact values; and $\Phi$ is uhc.

To prove the existence of a function of maximal dependence, it suffices to construct a measurable selector $\tilde\varphi$ of $\Phi$ such that $\lim_{u \to 0} \tilde \varphi(u) = 0$ and $\lim_{u \to 0} u^2/{\tilde \varphi}(u) = 0$.
According to the Kuratowski-Ryll-Nardzewski selection theorem (Theorem 18.13), a sufficient condition for $\Phi$ to admit a measurable selector is that $\Phi$ is weakly measurable and has non-empty closed values.
Since we know that $\Phi$ has non-empty closed values, it remains to show that $\Phi$ is weakly measurable.
As $B=[0,1]$ is compact, every closed set $F\subseteq B$ is compact. Since $\Phi$ is uhc, Part~3 of Lemma~17.4 yields that the lower inverse image $\Phi^\ell(F)$ is closed.
Since $A$ is equipped with the Borel $\sigma$-algebra, closed sets are measurable. Therefore, we have that $\Phi$ is measurable.
Since $B$ is metrizable, it follows from Lemma 18.2 that the measurability of $\Phi$ implies weak measurability.
Therefore, $\Phi$ admits a measurable selector, denoted by $\tilde \varphi$.

Next, we establish the asymptotic behavior of $\tilde \varphi$.
Since $\Lambda$ is non-degenerate, there exists $(x_0,y_0)\in(0,\infty)^2$ such that $\Lambda(x_0,y_0)>0$. By the homogeneity of tail copulas, we have $\Lambda(x_0,y_0)
=
\sqrt{x_0y_0}\,
\Lambda\left(\sqrt{x_0/y_0},\sqrt{y_0/x_0}\right)$.
Hence, with $b_0:=\sqrt{x_0/y_0}$,
we have $\Lambda(b_0,1/b_0)>0$. By the definition of the tail copula, there exist constants $\delta\in(0,1]$ and $u_0\in(0,1]$ such that $C\left(b_0u,u/b_0\right)\ge \delta u$ for
$0<u\le u_0$.
By decreasing $u_0$ if necessary, we may assume that
$u_0\le \min(b_0,1/b_0)$,
so that $b_0u\in [u^2,1]$ for all $0<u\le u_0$. Since $\tilde\varphi(u)\in\Phi(u)$, we obtain
\begin{align*}
\min\left(\tilde\varphi(u),\frac{u^2}{\tilde\varphi(u)}\right)
\ge
C\left(\tilde\varphi(u),\frac{u^2}{\tilde\varphi(u)}\right)
\ge
C\left(b_0u,\frac{u}{b_0}\right)
\ge
\delta u,
\qquad 0<u\le u_0,
\end{align*}
where the first inequality follows from the Fr\'echet--Hoeffding upper bound. Therefore,
\begin{align*}
\delta u\le \tilde\varphi(u)\le \frac{u}{\delta}
\qquad\text{and}\qquad
\delta u\le \frac{u^2}{\tilde\varphi(u)}\le \frac{u}{\delta},
\qquad 0<u\le u_0.
\end{align*}
Hence $\lim_{u\downarrow 0}\tilde\varphi(u)=0$ and $\lim_{u\downarrow 0}u^2/\tilde\varphi(u)=0$.
Thus $\tilde\varphi\in\mathcal A$. Now let $\varphi\in\mathcal A$ and $u\in(0,1]$.
Since $\varphi(u)\in [u^2,1]=\Gamma(u)$ by admissibility and $\tilde\varphi(u)\in\Phi(u)$, we have
\begin{align*}
C\left(\varphi(u),\frac{u^2}{\varphi(u)}\right)
\le
\max_{x\in\Gamma(u)} C\left(x,\frac{u^2}{x}\right)
=
C\left(\tilde\varphi(u),\frac{u^2}{\tilde\varphi(u)}\right).
\end{align*}
Therefore, $\tilde\varphi$ is a function of maximal dependence. Setting $\varphi^\ast:=\tilde\varphi$, we conclude that $C$ admits a function of maximal dependence.

\item Let $\varphi^\ast\in\mathcal A$ be any function of maximal dependence, whose existence is guaranteed by part~\ref{item:varphi:ast:existence}. For fixed $b\in(0,\infty)$, let $\delta_b:=\min(b,1/b)\in(0,1]$ and define $\varphi_b:(0,1]\to[0,1]$ by
\begin{align*}
\varphi_b(u):=
\begin{cases}
bu, & \text{if } 0<u\le \delta_b,\\
\max(u^2,b\delta_b), & \text{if } \delta_b<u\le 1.
\end{cases}
\end{align*}
It is straightforward to check that $\varphi_b\in\mathcal A$. Since $\varphi_b(u)=bu$ for $0<u\le \delta_b$, the maximality of $\varphi^\ast$ yields
\begin{align*}
\frac{C(\varphi^\ast(u),u^2/\varphi^\ast(u))}{u}
\ge
\frac{C(\varphi_b(u),u^2/\varphi_b(u))}{u}
=
\frac{C(bu,u/b)}{u},
\qquad 0<u\le \delta_b.
\end{align*}
Taking $\liminf_{u\downarrow 0}$, we obtain
\begin{align*}
\liminf_{u\downarrow 0}\frac{C(\varphi^\ast(u),u^2/\varphi^\ast(u))}{u}
\ge
\lim_{u\downarrow 0}\frac{C(bu,u/b)}{u}
=
\Lambda\left(b,\frac{1}{b}\right).
\end{align*}
Since $b>0$ was arbitrary,
\begin{align}
\liminf_{u\downarrow 0}\frac{C(\varphi^\ast(u),u^2/\varphi^\ast(u))}{u}
\ge
\sup_{b\in(0,\infty)}\Lambda\left(b,\frac{1}{b}\right)
=
\lambda^\ast.
\label{eq:liminf:bound}
\end{align}

We now extend the domain of $C$ by $$\bar C(x,y):=
C\bigl((x\vee 0)\wedge 1,(y\vee 0)\wedge 1\bigr),\qquad (x,y)\in\mathbb R^2,$$
which is simply the joint distribution function of $(U,V)\sim C$.
To simplify notation, we write $C$ in place of $\bar C$ below. For $u\in(0,1]$ and $b\in(0,\infty)$, define
\begin{align*}
f_u(b):=\frac{C(ub,u/b)}{u}
\qquad\text{and}\qquad
f(b):=\Lambda\left(b,\frac{1}{b}\right).
\end{align*}
Fix $L>1$ and set $K_L=[1/L,L]$.
Since $\{(b,1/b):b\in K_L\}$ is a compact subset of $(0,\infty)^2$, the local uniformity of the convergence defining the tail copula implies $\sup_{b\in K_L}|f_u(b)-f(b)|\to 0$ as $u\downarrow 0$; see \citet[Theorem~1(v)]{schmidt2006non}. Consequently,
\begin{align}
\lim_{u\downarrow 0}\sup_{b\in K_L}f_u(b)=\sup_{b\in K_L}f(b).
\label{eq:lim:sup:interchange}
\end{align}

Moreover, by the Fr\'echet--Hoeffding upper bound for $C$ and $\Lambda$, we have $C(x,y)\le \min(x,y)$ and $
\Lambda(x,y)\le \min(x,y)$, $(x,y) \in(0,\infty)^2$; see \citet[Theorem~2(i)]{schmidt2006non}. Hence, for every $u\in(0,1]$ and $b>0$, we obtain
$f_u(b)\le \min\left(b,1/b\right)$ and $
f(b)\le \min\left(b,1/b\right)$.
Therefore,
\begin{align}
\sup_{b\notin K_L}f_u(b)\le \frac1L
\qquad\text{and}\qquad
\sup_{b\notin K_L}f(b)\le \frac1L.
\label{eq:lim:sup:interchange:2}
\end{align}
Using
\begin{align*}
\sup_{b\in(0,\infty)}f_u(b)
=
\max\left\{\sup_{b\in K_L}f_u(b),\sup_{b\notin K_L}f_u(b)\right\},
\end{align*}
together with \eqref{eq:lim:sup:interchange} and \eqref{eq:lim:sup:interchange:2}, we obtain
\begin{align}
\limsup_{u\downarrow 0}\sup_{b\in(0,\infty)}f_u(b)
\le
\max\left\{\sup_{b\in K_L}f(b),\frac1L\right\}.
\label{eq:limsup}
\end{align}
On the other hand,
\begin{align}
\liminf_{u\downarrow 0}\sup_{b\in(0,\infty)}f_u(b)
\ge
\liminf_{u\downarrow 0}\sup_{b\in K_L}f_u(b)
=
\sup_{b\in K_L}f(b).
\label{eq:liminf}
\end{align}
Since $K_L\uparrow(0,\infty)$ as $L\to\infty$, we have
$
\sup_{b\in K_L}f(b)\uparrow \sup_{b\in(0,\infty)}f(b).
$
Moreover, by \eqref{eq:lim:sup:interchange:2}, $f(b)\to 0$ as $b\downarrow 0$ or $b\uparrow\infty$. Hence the right-hand side of \eqref{eq:limsup} converges to $\sup_{b\in(0,\infty)}f(b)$ as $L\to\infty$. Combining \eqref{eq:limsup} and \eqref{eq:liminf}, we conclude that
\begin{align}
\lim_{u\downarrow 0}\sup_{b\in(0,\infty)}\frac{C(ub,u/b)}{u}
=
\sup_{b\in(0,\infty)}\Lambda\left(b,\frac{1}{b}\right)
=
\lambda^\ast.
\label{eq:sub:goal}
\end{align}

For each $u\in(0,1]$, let
$
b_u^\ast:=\varphi^\ast(u)/u.
$
Since $\varphi^\ast(u)\in[u^2,1]$, we have $b_u^\ast\in[u,1/u]\subset(0,\infty)$, and therefore
$$
\frac{C(\varphi^\ast(u),u^2/\varphi^\ast(u))}{u}
=
\frac{C(ub_u^\ast,u/b_u^\ast)}{u}
\le
\sup_{b\in(0,\infty)}\frac{C(ub,u/b)}{u}.
$$
Taking $\limsup_{u\downarrow 0}$ and using \eqref{eq:sub:goal}, we obtain
\begin{align}
\limsup_{u\downarrow 0}\frac{C(\varphi^\ast(u),u^2/\varphi^\ast(u))}{u}
\le
\lambda^\ast.
\label{eq:limsup:bound}
\end{align}
Combining \eqref{eq:liminf:bound} and \eqref{eq:limsup:bound}, we find
\begin{align*}
\lambda^\ast
\le
\liminf_{u\downarrow 0}\frac{C(\varphi^\ast(u),u^2/\varphi^\ast(u))}{u}
\le
\limsup_{u\downarrow 0}\frac{C(\varphi^\ast(u),u^2/\varphi^\ast(u))}{u}
\le
\lambda^\ast.
\end{align*}
Hence the limit $\lambda_{\varphi^\ast}(C)$ exists and satisfies
$
\lambda_{\varphi^\ast}(C)=\lambda^\ast(C).
$

\item Assume that the supremum of
$
f(b)=\Lambda\left(b,1/b\right)$
is uniquely attained at $b^\ast\in(0,\infty)$. Let $\varphi^\ast\in\mathcal A$ be a function of maximal dependence and set
$b_u^\ast:=\varphi^\ast(u)/u$, $u\in(0,1]$.
Since $\varphi^\ast(u)\in [u^2,1]$, we have $b_u^\ast\in[u,1/u]$. By the change of variable $t=bu$, the maximality of $\varphi^\ast(u)$ implies that
$b_u^\ast\in \operatorname*{argmax}_{b\in[u,1/u]} f_u(b)$,
$u\in(0,1]$.

Fix $\varepsilon>0$. Since $f$ is continuous on $(0,\infty)$, $f(b^\ast)=\lambda^\ast>0$, and $f(b)\to 0$ as $b\downarrow 0$ or $b\uparrow\infty$, we can choose $L>1$ and $\delta\in(0,f(b^\ast)/2)$ such that
$
b^\ast\in K_L=[1/L,L]$ and
$ 1/L\le f(b^\ast)-2\delta$.
Then, by \eqref{eq:lim:sup:interchange:2},
\begin{align}
\sup_{b\notin K_L}f_u(b)\le \frac1L\le f(b^\ast)-2\delta,
\qquad u\in(0,1].
\label{eq:outside:small}
\end{align}
By the local uniform convergence on $K_L$, there exists $u_{L,\delta}\in(0,1/L)$ such that
\begin{align}
\sup_{b\in K_L}|f_u(b)-f(b)|\le \delta,
\qquad 0<u<u_{L,\delta}.
\label{eq:uniform:compact}
\end{align}

We claim that $b_u^\ast\in K_L$ for every $0<u<u_{L,\delta}$. Indeed, if $b_u^\ast\notin K_L$, then by \eqref{eq:outside:small}, we have
$f_u(b_u^\ast)\le 1/L\le f(b^\ast)-2\delta$.
On the other hand, since $b^\ast\in K_L\subset [u,1/u]$ for $u<1/L$, \eqref{eq:uniform:compact} gives
$
f_u(b^\ast)\ge f(b^\ast)-\delta.
$
Hence
$
f_u(b_u^\ast)\le f(b^\ast)-2\delta < f(b^\ast)-\delta\le f_u(b^\ast),
$
which contradicts the maximality of $b_u^\ast$ over $[u,1/u]$. Thus
$
b_u^\ast\in K_L$ for all $0<u<u_{L,\delta}$.

Now set
\begin{align*}
F_{L,\varepsilon}:=
K_L\cap\{b\in(0,\infty): |b-b^\ast|\ge \varepsilon\}.
\end{align*}
If $F_{L,\varepsilon}=\varnothing$, then $b_u^\ast\in K_L$ already implies
$
|b_u^\ast-b^\ast|<\varepsilon$ for all $0<u<u_{L,\delta}$, and there is nothing more to prove. Assume therefore that $F_{L,\varepsilon}\neq\varnothing$. Since $F_{L,\varepsilon}$ is compact and $f$ has a unique maximizer at $b^\ast$, there exists $\eta>0$ such that
$
\sup_{b\in F_{L,\varepsilon}} f(b)\le f(b^\ast)-3\eta.
$
Again by the local uniform convergence on $K_L$, there exists $u_{L,\eta}\in(0,1/L)$ such that
\begin{align}
\sup_{b\in K_L}|f_u(b)-f(b)|\le \eta,
\qquad 0<u<u_{L,\eta}.
\label{eq:local:unif:KL}
\end{align}
Set
$
u_\varepsilon:=\min(u_{L,\delta},u_{L,\eta})
$
and let $0<u<u_\varepsilon$. We already know that $b_u^\ast\in K_L$. Suppose, toward a contradiction, that
$
|b_u^\ast-b^\ast|\ge \varepsilon.
$
Then $b_u^\ast\in F_{L,\varepsilon}$, and therefore, by the choice of $\eta$ and \eqref{eq:local:unif:KL}, we have
$f_u(b_u^\ast)\le f(b_u^\ast)+\eta\le f(b^\ast)-2\eta.
$
Since $b^\ast\in K_L\subset [u,1/u]$ and \eqref{eq:local:unif:KL} also yields
$
f_u(b^\ast)\ge f(b^\ast)-\eta,
$
we obtain
$
f_u(b_u^\ast)\le f(b^\ast)-2\eta < f(b^\ast)-\eta\le f_u(b^\ast),
$
again contradicting the maximality of $b_u^\ast$ over $[u,1/u]$. Therefore
$
|b_u^\ast-b^\ast|<\varepsilon$ for all $0<u<u_\varepsilon$.
Since $\varepsilon>0$ was arbitrary, we conclude that
$
\lim_{u\downarrow 0}\varphi^\ast(u)/u=b^\ast$.
\qedhere
\end{enumerate}
\end{proof}

\subsection{Lemma~\ref{lem:mtcm:and:h}}

\begin{proof}[Proof of Lemma~\ref{lem:mtcm:and:h}]
\hspace{0mm}
 \begin{enumerate}[label=(\roman*), labelwidth=\widthof{(ii)}]
\item We consider the following form of the MTCM by changing the variable $b=e^s$:
\begin{align*}
    \lambda^\ast(\hat C)=\max_{b \in (0,\infty)}\Lambda\left(b,\frac{1}{b};\hat C\right)
    =\max_{s \in \R}\Lambda\left(e^s,e^{-s};\hat C\right).
\end{align*}
If $H$ admits a density on $(0,1)$, then~\eqref{eq:Lambda:H} yields
 \begin{align*}
  \Lambda(x,y;\hat C)=\int_0^1 \min\{wx,(1-w)y\}h(w) \,{\rm d}w,
  \quad x,y \in (0,\infty).
\end{align*}
Note that point masses of $H$ at $0$ and $1$ are not relevant since $\min(0\cdot x,1\cdot y)=\min(1\cdot x,0\cdot y)=0$ for any $x,y \in (0,\infty)$.

For $w\in(0,1)$, define $a(w)=(1/2)\,\{\ln w - \ln (1-w)\}$.
Then $w=w(a)=e^{2a}/(1+e^{2a})$ and $(1-w)/w=e^{-2a}$, hence ${\rm d}w=2w(1-w)\,{\rm d}a$.
Moreover, we have that
\begin{align*}
  \min\!\left(w e^s,(1-w)e^{-s}\right)=\sqrt{w(1-w)}\,e^{-|s+a|}
\end{align*}
since $we^s=\sqrt{w(1-w)}\,e^{s+a}$ and $(1-w)e^{-s}=\sqrt{w(1-w)}\,e^{-(s+a)}$.
Therefore,
\begin{align*}
  \Lambda(e^s,e^{-s})
  &=\int_0^1 \min\!\left(e^s w,e^{-s}(1-w)\right)h(w)\,dw \\
  &=\int_{-\infty}^{\infty} \sqrt{w(a)(1-w(a))}\,e^{-|s+a|}\,h(w(a))\,2w(a)(1-w(a))\,\rd a \\
  &=\int_{-\infty}^{\infty} e^{-|s+a|}\,m(a)\,\rd a,
\end{align*}
where $m$ is defined by~\eqref{eq:m:function}.
\item For convenience, write the function $L(s)=\int_{-\infty}^\infty e^{-|s+a|}m(a)\,\rd a$.
If $m$ is even, it is straightforward to check that $L$ is also even.
Therefore, if, in addition, $L$ is strictly decreasing on $(0,\infty)$, then $L$ attains its unique maximum at $s=0$ (that is $b^\ast=1)$.

Set $k(x)=e^{-|x|}$ so that $L(s)=\int_{-\infty}^{\infty} k(s+a)m(a)\,\rd a$ for $s>0$.
Since $k$ is globally Lipschitz with Lipschitz constant $1$, for every $h\neq 0$
and every $a\in\R$, we have
\begin{align*}
  \left|
  \frac{k(s+h+a)-k(s+a)}{h}
  \right|
  \le 1.
\end{align*}
Moreover, for every $a\neq -s$, the function $k$ is differentiable at $s+a$, and
\begin{align*}
  \lim_{h\to 0}
  \frac{k(s+h+a)-k(s+a)}{h}
  =
  -\operatorname{sgn}(s+a)e^{-|s+a|}.
\end{align*}
Since the exceptional set $\{-s\}$ has Lebesgue measure zero and $|m|$ is
integrable on $\R$, the dominated convergence theorem yields
\begin{align*}
  L'(s)
  &=
  \int_{-\infty}^{\infty}
  \lim_{h\to 0}
  \frac{k(s+h+a)-k(s+a)}{h}\,m(a)\,\rd a \\
  &=
  -\int_{-\infty}^{\infty}
  \operatorname{sgn}(s+a)e^{-|s+a|}\,m(a)\,\rd a.
\end{align*}
Then, for $s>0$,
\begin{align*}
  L'(s)&=-\int_0^\infty 1\cdot e^{-(s+a)}m(a)\rd a  - \int_0^\infty \operatorname{sgn}(s-a)\,e^{-|s-a|}m(a)\,\rd a \\
  &=-\int_0^s \left\{e^{-(s+a)}+e^{-(s-a)}\right\}m(a)\,\rd a
    -\int_s^\infty \left\{e^{-(s+a)}-e^{-(a-s)}\right\}m(a)\,\rd a \\
  &=-2e^{-s}\int_0^s \cosh(a)\,m(a)\,\rd a
    +2\sinh(s)\int_s^\infty e^{-a}\,m(a)\,\rd a.
\end{align*}
Therefore, by using strict decreasingness of $m$, we have that
\begin{align*}
 L'(s)
  &< -2e^{-s}m(s)\int_0^s \cosh(a)\,\rd a
  + 2\sinh(s)m(s)\int_s^\infty e^{-a}\,\rd a\\
  &=-2e^{-s}m(s)\sinh(s) + 2\sinh(s)m(s)e^{-s}\\
  &=0,
\end{align*}
which completes the proof.\qedhere
\end{enumerate}
\end{proof}

\subsection{Proposition~\ref{prop:bivariate:t:b:1}}

We first derive the density of the spectral measure of the $t$-EV copula.
Denote by $H_{\nu,\rho}$ the spectral measure of the bivariate $t$-EV copula~\citep{demarta2005t,nikoloulopoulos2009extreme} for the correlation parameter $\rho\in (-1,1)$ and the degrees of freedom $\nu>0$.
Below, $T_{\nu+1}$ and $t_{\nu+1}$ denote the cdf and density, respectively, of the univariate
Student $t$ distribution with $\nu+1$ degrees of freedom.
We need the following lemma to prove Proposition~\ref{prop:bivariate:t:b:1}.

\begin{Lemma}[Spectral measure of bivariate $t$-EV copulas]\label{lem:bivariate:t:h}
For the bivariate $t$-EV copula with correlation parameter
$\rho\in(-1,1)$ and degrees of freedom $\nu>0$, define
\begin{align*}
  \eta=\sqrt{\frac{\nu+1}{1-\rho^2}}
  \qquad\text{and}\qquad
  r(w)=\frac{1-w}{w},\qquad w\in(0,1).
\end{align*}
Then the following statements hold.
\begin{enumerate}[label=(\roman*), labelwidth=\widthof{(iii)}]
\item The restriction of $H_{\nu,\rho}$ to $(0,1)$ is absolutely continuous with
Lebesgue density
\begin{align*}
  h_{\nu,\rho}(w)
  =
  \frac{\eta}{\nu}\,
  \frac{
    t_{\nu+1}\bigl(\eta\{r(w)^{1/\nu}-\rho\}\bigr)
  }{
    w^{(2\nu+1)/\nu}(1-w)^{(\nu-1)/\nu}
  },
  \quad w\in(0,1).
\end{align*}
\item The endpoint masses are $H_{\nu,\rho}(\{0\})=H_{\nu,\rho}(\{1\})=T_{\nu+1}(-\eta\rho)$.
\item The density is symmetric, i.e.,\ $h_{\nu,\rho}(w)=h_{\nu,\rho}(1-w)$ for every
$w\in(0,1)$.
\item The interior mass satisfies
\begin{align*}
  \int_0^1 h_{\nu,\rho}(w)\,\mathrm{d}w
  =
  2-H_{\nu,\rho}(\{0\})-H_{\nu,\rho}(\{1\})
  =
  2\,T_{\nu+1}(\eta\rho).
\end{align*}
\end{enumerate}
\end{Lemma}

\begin{proof}
By \citet[Theorem~2.3 with $d=2$]{nikoloulopoulos2009extreme}, the tail copula of the survival $t$-EV copula is
\begin{align}\label{eq:tEV:Lambda:explicit}
  \Lambda(x,y)
  =
  x\,T_{\nu+1}\!\left(\eta\left[\rho-\left(\frac{y}{x}\right)^{-1/\nu}\right]\right)
  +
  y\,T_{\nu+1}\!\left(\eta\left[\rho-\left(\frac{x}{y}\right)^{-1/\nu}\right]\right),
  \quad (x,y)\in(0,\infty)^2.
\end{align}
Moreover, as in Lemma~\ref{lem:mtcm:and:h}, the spectral representation gives
\begin{align}\label{eq:tEV:Lambda:spectral}
  \Lambda(x,y)
  =
  \int_{[0,1]} \min\{wx,(1-w)y\}\,H_{\nu,\rho}(\mathrm{d}w),
  \quad (x,y)\in(0,\infty)^2.
\end{align}

\begin{enumerate}[label=(\roman*), labelwidth=\widthof{(iii)}]
\item
Fix $x,y>0$ and set $a=y/(x+y)\in(0,1)$. Define
\begin{align*}
  G(t)=\int_{[0,t]} w\,H_{\nu,\rho}(\mathrm{d}w),\qquad t\in[0,1].
\end{align*}
For fixed $y>0$ and $w\in[0,1]$, the map $x\mapsto \min\{wx,(1-w)y\}$ is Lipschitz with constant $w$. Hence, for every $h\neq 0$,
\begin{align*}
  \left|
  \frac{
    \min\{w(x+h),(1-w)y\}-\min\{wx,(1-w)y\}
  }{h}
  \right|
  \le w.
\end{align*}
Since $H_{\nu,\rho}$ is finite (indeed, its total mass equals $2$ by the moment
constraints), the dominated convergence theorem applied to
\eqref{eq:tEV:Lambda:spectral} yields the one-sided derivatives
\begin{align*}
  \frac{\partial^+ \Lambda}{\partial x}(x,y)
  &=
  \int_{[0,a)} w\,H_{\nu,\rho}(\mathrm{d}w),\\
  \frac{\partial^- \Lambda}{\partial x}(x,y)
  &=
  \int_{[0,a]} w\,H_{\nu,\rho}(\mathrm{d}w)
  =
  G(a).
\end{align*}
Indeed, the pointwise limit of the difference quotient is $w \id_{\{w<a\}}$ when $h\downarrow 0$, and is $w \id_{\{w\le a\}}$ when $h\uparrow 0$.
On the other hand, we have from \eqref{eq:tEV:Lambda:explicit} that $\Lambda$ is of
class $\operatorname{C}^2$ on $(0,\infty)^2$, so the derivative with respect to
$x$ exists. Therefore, we obtain
\begin{align*}
  a\,H_{\nu,\rho}(\{a\})
  =
  \frac{\partial^- \Lambda}{\partial x}(x,y)
  -
  \frac{\partial^+ \Lambda}{\partial x}(x,y)
  =0.
\end{align*}
Since $a\in(0,1)$ was arbitrary, $H_{\nu,\rho}$ has no atoms in $(0,1)$, and
hence
\begin{align}\label{eq:formula:Lambda:x:G}
\frac{\partial}{\partial x} \Lambda(x,y)=G(a),\quad\text{where}\quad a=\frac{y}{x+y}.
\end{align}

Now fix $x>0$ and define
\begin{align*}
  F_x(t)=\frac{\partial \Lambda}{\partial x}\left(x,\frac{xt}{1-t}\right),
  \qquad t\in(0,1).
\end{align*}
Because $\Lambda$ is of class $\operatorname{C}^2$, the map $F_x$ is of
class $\operatorname{C}^1$. By~\eqref{eq:formula:Lambda:x:G}, we have $F_x(t)=G(t)$ for every
$t\in(0,1)$, so $G$ is differentiable on $(0,1)$. By the chain rule,
\begin{align*}
  G'(t)
  =
  F_x'(t)
  =
  \frac{\partial^2 \Lambda}{\partial x\,\partial y}
  \left(x,\frac{xt}{1-t}\right)\frac{x}{(1-t)^2},
  \qquad t\in(0,1).
\end{align*}
Evaluating this at $x=1-w$ and $t=w$ gives
\begin{align*}
  G'(w)=\frac{1}{1-w}\,\frac{\partial^2 \Lambda}{\partial x\,\partial y}(1-w,w),
  \qquad w\in(0,1).
\end{align*}
Since $G(t)=\mu([0,t])$, $t\in(0,1)$, where
$\mu(\mathrm dw):=w\,H_{\nu,\rho}(\mathrm dw)$, we have that, for every
$0<a<b<1$,
\begin{align*}
\mu((a,b])=G(b)-G(a)=\int_a^b G'(w)\,\mathrm dw.
\end{align*}
Hence $\mu|_{(0,1)}$ is absolutely continuous with respect to Lebesgue
measure with density $G'$. Since $\mu(\mathrm dw):=w\,H_{\nu,\rho}(\mathrm dw)$ and $w>0$ on
$(0,1)$, it follows that $H_{\nu,\rho}|_{(0,1)}$ is also absolutely
continuous. 
Writing $h_{\nu,\rho}$ as the density of $H_{\nu,\rho}$, we obtain $G'(w)=w\,h_{\nu,\rho}(w)$ for a.e. $w\in(0,1)$.
Therefore,
\begin{align}\label{eq:tEV:h:inversion}
  h_{\nu,\rho}(w)
  =
  \frac{G'(w)}{w}
  =
  \frac{1}{w(1-w)}\,\frac{\partial^2 \Lambda}{\partial x\,\partial y}(1-w,w),
  \quad \text{for a.e. }w\in(0,1).
\end{align}
Since the right-hand side is continuous on \((0,1)\), we may take this
continuous version as the density.
Let $x,y>0$ and set $u=(y/x)^{-1/\nu}$ and $v=(x/y)^{-1/\nu}=u^{-1}$. A direct
differentiation of \eqref{eq:tEV:Lambda:explicit} gives
\begin{align*}
  \frac{\partial \Lambda}{\partial x}(x,y)
  =
  T_{\nu+1}\bigl(\eta(\rho-u)\bigr)
  -
  \frac{\eta}{\nu}\,u\,t_{\nu+1}\bigl(\eta(\rho-u)\bigr)
  +
  \frac{\eta}{\nu}\,u^{-(\nu+1)}\,t_{\nu+1}\bigl(\eta(\rho-u^{-1})\bigr).
\end{align*}
We use the identity
\begin{align}\label{eq:tEV:t-identity}
  t_{\nu+1}\bigl(\eta(\rho-u^{-1})\bigr)=u^{\nu+2}\,t_{\nu+1}\bigl(\eta(\rho-u)\bigr),
  \quad u>0,
\end{align}
which follows by direct algebra from the explicit Student $t$ density together
with $\eta^2=(\nu+1)/(1-\rho^2)$. Substituting \eqref{eq:tEV:t-identity} yields
cancellation of the last two terms and thus $(\partial/\partial x) \Lambda(x,y)
  =
  T_{\nu+1}\bigl(\eta(\rho-u)\bigr)$.
Differentiating with respect to $y$ gives
\begin{align*}
  \frac{\partial^2 \Lambda}{\partial x\,\partial y}(x,y)
  =
  t_{\nu+1}\bigl(\eta(\rho-u)\bigr)\,\eta\,\frac{\partial(\rho-u)}{\partial y}
  =
  t_{\nu+1}\bigl(\eta(\rho-u)\bigr)\,\eta\,
  \frac{1}{\nu}\left(\frac{y}{x}\right)^{-(\nu+1)/\nu}\frac{1}{x}.
\end{align*}
Now evaluate at $(x,y)=(1-w,w)$, where $u=r(w)^{1/\nu}$. Using
$t_{\nu+1}(z)=t_{\nu+1}(-z)$ and \eqref{eq:tEV:h:inversion}, we obtain
\begin{align*}
  h_{\nu,\rho}(w)
  =
  \frac{1}{w(1-w)}\,\frac{\partial^2 \Lambda}{\partial x\,\partial y}(1-w,w)
  =
  \frac{\eta}{\nu}\,
  \frac{
    t_{\nu+1}\bigl(\eta\{r(w)^{1/\nu}-\rho\}\bigr)
  }{
    w^{(2\nu+1)/\nu}(1-w)^{(\nu-1)/\nu}
  }.
\end{align*}
\item For each fixed $w\in[0,1]$, the map $y\mapsto \min\{w,(1-w)y\}$ is increasing and
\begin{align*}
  \lim_{y\to\infty}\min(w,(1-w)y)
  =
  \begin{cases}
    w, & w\in[0,1),\\
    0, & w=1.
  \end{cases}
\end{align*}
Hence, by the monotone convergence theorem applied to \eqref{eq:tEV:Lambda:spectral},
\begin{align}\label{eq:tEV:limit:spectral}
  \lim_{y\to\infty}\Lambda(1,y)
  =
  \int_{[0,1)} w\,H_{\nu,\rho}(\mathrm{d}w)
  =
  \int_{[0,1]} w\,H_{\nu,\rho}(\mathrm{d}w)-H_{\nu,\rho}(\{1\})
  =
  1-H_{\nu,\rho}(\{1\}).
\end{align}
On the other hand, from \eqref{eq:tEV:Lambda:explicit}, we have
\begin{align*}
  \Lambda(1,y)
  =
  T_{\nu+1}\Bigl(\eta\bigl[\rho-y^{-1/\nu}\bigr]\Bigr)
  +
  y\,T_{\nu+1}\Bigl(\eta\bigl[\rho-y^{1/\nu}\bigr]\Bigr).
\end{align*}
The first term converges to $T_{\nu+1}(\eta\rho)$ as $y\to\infty$.
For the second term, note that there exists $c>0$ such that
$t_{\nu+1}(z)\le c z^{-(\nu+2)}$ for all $z\ge 1$, hence
\begin{align*}
  T_{\nu+1}(-z)=\int_z^\infty t_{\nu+1}(u)\,\mathrm{d}u
  \le \frac{c}{\nu+1}\,z^{-(\nu+1)},
  \quad z\ge 1.
\end{align*}
With $z=\eta(y^{1/\nu}-\rho)\to\infty$, this yields
$y\,T_{\nu+1}(\eta[\rho-y^{1/\nu}])\to 0$. Therefore,
\begin{align}\label{eq:tEV:limit:explicit}
  \lim_{y\to\infty}\Lambda(1,y)=T_{\nu+1}(\eta\rho).
\end{align}
Comparing \eqref{eq:tEV:limit:spectral} and \eqref{eq:tEV:limit:explicit} gives $H_{\nu,\rho}(\{1\})=1-T_{\nu+1}(\eta\rho)=T_{\nu+1}(-\eta\rho)$.
The identity $H_{\nu,\rho}(\{0\})=T_{\nu+1}(-\eta\rho)$ follows analogously by
considering $\Lambda(x,1)$ as $x\to\infty$.

\item
Fix $w\in(0,1)$ and set $u=r(w)^{1/\nu}$.
Since $r(1-w)=1/r(w)$, we have $r(1-w)^{1/\nu}=1/u$. Using the explicit density from (i),
\begin{align*}
  \frac{h_{\nu,\rho}(1-w)}{h_{\nu,\rho}(w)}
  =
  \frac{
    t_{\nu+1}\bigl(\eta(u^{-1}-\rho)\bigr)
  }{
    t_{\nu+1}\bigl(\eta(u-\rho)\bigr)
  }\,
  \left(\frac{w}{1-w}\right)^{(\nu+2)/\nu}.
\end{align*}
Since $u^\nu=r(w)=(1-w)/w$, we have $(w/(1-w))^{(\nu+2)/\nu}=u^{-(\nu+2)}$.
Combining this with \eqref{eq:tEV:t-identity} yields $h_{\nu,\rho}(1-w)/h_{\nu,\rho}(w)
  =
  u^{\nu+2}\,u^{-(\nu+2)}=1$, hence $h_{\nu,\rho}(1-w)=h_{\nu,\rho}(w)$.

\item
The first equation directly follows from the moment constraints on the spectral measure, and the second equation is an immediate consequence from (ii).
In the following, we show the first equation by direct calculation.

Set $u=r(w)^{1/\nu}=((1-w)/w)^{1/\nu}\in(0,\infty)$. Then
\begin{align*}
  w=\frac{1}{1+u^\nu},
  \quad
  1-w=\frac{u^\nu}{1+u^\nu}
  \quad\text{and}\quad
  \mathrm{d}w=-\frac{\nu u^{\nu-1}}{(1+u^\nu)^2}\,\mathrm{d}u.
\end{align*}
Moreover, a direct simplification gives
\begin{align*}
  \frac{\mathrm{d}w}{w^{(2\nu+1)/\nu}(1-w)^{(\nu-1)/\nu}}
  =
  -\,\nu(1+u^\nu)\,\mathrm{d}u.
\end{align*}
Sending $w\downarrow 0$ and $w\uparrow 1$ yield $ u\uparrow\infty$ and
$u\downarrow 0$, respectively, and we obtain
\begin{align*}
  \int_0^1 h_{\nu,\rho}(w)\,\mathrm{d}w
  &=
  \int_\infty^0
  \frac{\eta}{\nu}\,
  t_{\nu+1}\bigl(\eta(u-\rho)\bigr)\,
  \frac{\mathrm{d}w}{w^{(2\nu+1)/\nu}(1-w)^{(\nu-1)/\nu}} \\
  &=
  \eta\int_0^\infty (1+u^\nu)\,t_{\nu+1}\bigl(\eta(u-\rho)\bigr)\,\mathrm{d}u \\
  &=
  \eta\int_0^\infty t_{\nu+1}\bigl(\eta(u-\rho)\bigr)\,\mathrm{d}u
  +
  \eta\int_0^\infty u^\nu\,t_{\nu+1}\bigl(\eta(u-\rho)\bigr)\,\mathrm{d}u \\
  &=:I_1+I_2.
\end{align*}
For $I_1$, substituting $z=\eta(u-\rho)$ leads to
\begin{align*}
  I_1
  =
  \int_{-\eta\rho}^\infty t_{\nu+1}(z)\,\mathrm{d}z
  =
  1-T_{\nu+1}(-\eta\rho)
  =
  T_{\nu+1}(\eta\rho).
\end{align*}
For $I_2$, rearranging \eqref{eq:tEV:t-identity} as
$u^\nu t_{\nu+1}(\eta(u-\rho))=u^{-2}t_{\nu+1}(\eta(u^{-1}-\rho))$ and substituting
$v=u^{-1}$ yields
\begin{align*}
  I_2
  &=
  \eta\int_0^\infty u^{-2}\,t_{\nu+1}\bigl(\eta(u^{-1}-\rho)\bigr)\,\mathrm{d}u
  =
  \eta\int_\infty^0 t_{\nu+1}\bigl(\eta(v-\rho)\bigr)\,(-\mathrm{d}v)
  =
  I_1.
\end{align*}
Therefore $\int_0^1 h_{\nu,\rho}(w)\,\mathrm{d}w=I_1+I_2=2T_{\nu+1}(\eta\rho)$.\qedhere
\end{enumerate}
\end{proof}

We are now ready to prove Proposition~\ref{prop:bivariate:t:b:1}.

\begin{proof}[Proof of Proposition~\ref{prop:bivariate:t:b:1}]

Since the bivariate $t$-copula is radially symmetric, we have
$\hat C_{\nu,\rho}=C_{\nu,\rho}$ and hence $\lambda^\ast(C_{\nu,\rho})=\lambda^\ast(\hat C_{\nu,\rho})$.
Moreover, because $C_{\nu,\rho}$ is in the domain of attraction of the
corresponding $t$-EV copula, Lemma~\ref{lem:mtcm:and:h} applies with spectral
measure $H_{\nu,\rho}$. Therefore, it suffices to verify that the function $m$
in \eqref{eq:m:function} is even and integrable on $\mathbb{R}$ and is strictly
decreasing on $(0,\infty)$.

We first show that $m$ is even.
Let $w(a)=e^{2a}/(1+e^{2a})$. Then $w(-a)=1-w(a)$ and
$w(-a)\{1-w(-a)\}=w(a)\{1-w(a)\}$. By Lemma~\ref{lem:bivariate:t:h}~(iii) we have that
\begin{align*}
  m(-a)&=2\{w(-a)(1-w(-a))\}^{3/2}\,h_{\nu,\rho}(w(-a))\\
  &=2\{w(a)(1-w(a))\}^{3/2}\,h_{\nu,\rho}(1-w(a))\\
  &=2\{w(a)(1-w(a))\}^{3/2}\,h_{\nu,\rho}(w(a))\\
  &=m(a).
\end{align*}

Next, we show that $m$ is integrable.
Fix $a>0$ and set $w=w(a)$. Then $r(w)=(1-w)/w=e^{-2a}$ and
$r(w)^{1/\nu}=e^{-2a/\nu}$. Using Lemma~\ref{lem:bivariate:t:h}~(i) and the
definition of $m$ in \eqref{eq:m:function}, we obtain
\begin{align*}
  m(a)
  &=
  2\{w(1-w)\}^{3/2}\,h_{\nu,\rho}(w)\\
  &=
  2\{w(1-w)\}^{3/2}\,
  \frac{\eta}{\nu}\,
  \frac{
    t_{\nu+1}\bigl(\eta\{r(w)^{1/\nu}-\rho\}\bigr)
    }{
    w^{(2\nu+1)/\nu}(1-w)^{(\nu-1)/\nu}
  } \\
  &=
  \frac{2\eta}{\nu}\,
  w^{3/2-(2\nu+1)/\nu}\,
  (1-w)^{3/2-(\nu-1)/\nu}\,
  t_{\nu+1}\bigl(\eta\{e^{-2a/\nu}-\rho\}\bigr).
\end{align*}
Since $3/2-(2\nu+1)/\nu=-1/2-1/\nu$ and $3/2-(\nu-1)/\nu=1/2+1/\nu$, we have
\begin{align*}
  w^{3/2-(2\nu+1)/\nu}(1-w)^{3/2-(\nu-1)/\nu}
  =
  \left(\frac{1-w}{w}\right)^{1/2+1/\nu}
  =
  r(w)^{1/2+1/\nu}
  =
  e^{-(1+2/\nu)a}.
\end{align*}
Therefore, for $a>0$,
\begin{align}\label{eq:m:explicit:a>0}
  m(a)
  =
  \frac{2\eta}{\nu}\,e^{-(1+2/\nu)a}\,
  t_{\nu+1}\Bigl(\eta\{e^{-2a/\nu}-\rho\}\Bigr).
\end{align}
Since $t_{\nu+1}$ is bounded on $\mathbb{R}$, we have from~\eqref{eq:m:explicit:a>0} that
$m(a)\le K e^{-(1+2/\nu)a}$ for $a>0$ and some constant $K>0$. Hence
$\int_0^\infty m(a)\,\mathrm{d}a<\infty$.
Since $m$ is an even function, we conclude that $m$ is integrable on
$\mathbb{R}$.

Finally, we show that $m$ is strictly decreasing on $(0,\infty)$.
Fix $a>0$ and set $u(a)=e^{-2a/\nu}\in(0,1)$ and
$x(a)=\eta\{u(a)-\rho\}$. By~\eqref{eq:m:explicit:a>0}, we have 
$\ln m(a)=\mathrm{const}-(1+2/\nu)a+\ln t_{\nu+1}(x(a))$.
Since $u'(a)=-(2/\nu)u(a)$, we have $x'(a)=-(2\eta/\nu)u(a)$.
It also holds that
$(\mathrm{d}/\mathrm{d}x)\ln t_{\nu+1}(x) = -(\nu+2)x /(\nu+1+x^2)$
and thus, for $a>0$, that
\begin{align*}
  \frac{\mathrm{d}}{\mathrm{d}a}\ln m(a)
  &=
  -\left(1+\frac{2}{\nu}\right)
  -(\nu+2)\frac{x(a)}{\nu+1+x(a)^2}\,x'(a) \\
  &=
  -\left(1+\frac{2}{\nu}\right)
  +\frac{2(\nu+2)}{\nu}\,
  \frac{\eta^2 u(a)\{u(a)-\rho\}}{\nu+1+\eta^2\{u(a)-\rho\}^2}.
\end{align*}
Using $\eta^2=(\nu+1)/(1-\rho^2)$, we obtain $$\frac{\eta^2}{\nu+1+\eta^2(u-\rho)^2} = \frac{1}{1+u^2-2\rho u},$$
and hence
\begin{align*}
  \frac{\mathrm{d}}{\mathrm{d}a}\ln m(a)
  =
  -\frac{\nu+2}{\nu}
  +\frac{2(\nu+2)}{\nu}\,
  \frac{u(a)\{u(a)-\rho\}}{1+u(a)^2-2\rho u(a)}
  =
  \frac{\nu+2}{\nu}\,
  \frac{u(a)^2-1}{1+u(a)^2-2\rho u(a)}.
\end{align*}
Since $u(a)\in(0,1)$ for $a>0$, we have $u(a)^2-1<0$. Moreover, $1+u^2-2\rho u=(u-\rho)^2+(1-\rho^2)>0$
for all $u\in\mathbb{R}$ and $\rho\in(-1,1)$. Therefore, $ (\mathrm{d}/\mathrm{d}a)\ln m(a)<0$ for all $a>0$, which completes the proof.
\end{proof}

\subsection{Lemma~\ref{lem:tail:copula:smo}}

\begin{proof}[Proof of Lemma~\ref{lem:tail:copula:smo}]
\noindent For $(x,y)\in(0,\infty)^2$ and $t\downarrow0$, we have
\begin{align*}
  (1-tx)^{1-\alpha}=1-(1-\alpha)tx+o(t)
  \qquad\text{and}\qquad
  (1-ty)^{1-\beta}=1-(1-\beta)ty+o(t).
\end{align*}
Hence
\begin{align*}
  (1-ty)(1-tx)^{1-\alpha}
  &=1-\{y+(1-\alpha)x\}t+o(t),\\
  (1-tx)(1-ty)^{1-\beta}
  &=1-\{x+(1-\beta)y\}t+o(t).
\end{align*}
Using the identity $\hat C_{\alpha,\beta}^{\mathrm{MO}}(u,v)
  =u+v-1+C_{\alpha,\beta}^{\mathrm{MO}}(1-u,1-v)$, we obtain
\begin{align*}
  \Lambda(x,y;\hat C_{\alpha,\beta}^{\mathrm{MO}})
  &=
  \lim_{t\downarrow0}\frac{\hat C_{\alpha,\beta}^{\mathrm{MO}}(tx,ty)}{t}\\
  &=
  \lim_{t\downarrow0}
  \frac{tx+ty-1+C_{\alpha,\beta}^{\mathrm{MO}}(1-tx,1-ty)}{t}\\
  &=
  \lim_{t\downarrow0}
  \frac{tx+ty-1+\min\!\left((1-ty)(1-tx)^{1-\alpha},\,(1-tx)(1-ty)^{1-\beta}\right)}{t}\\
  &=
  x+y-\max\{y+(1-\alpha)x,\ x+(1-\beta)y\}\\
  &=
  \min(\alpha x,\beta y).
\end{align*}
Therefore $\Lambda(b,1/b;\hat C_{\alpha,\beta}^{\mathrm{MO}})
  =\min\!\left(\alpha b,\beta/b\right)$, $b>0$.
This function is strictly increasing on $(0,\sqrt{\beta/\alpha})$ and strictly decreasing on
$(\sqrt{\beta/\alpha},\infty)$. Hence the maximum is uniquely attained at $b^\ast=\sqrt{\beta/\alpha}$, and $\lambda^\ast(\hat C_{\alpha,\beta}^{\mathrm{MO}})=\sqrt{\alpha\beta}$.
\end{proof}

\subsection{Proposition~\ref{prop:smo}}

\begin{proof}[Proof of Proposition~\ref{prop:smo}]\mbox{}
  \begin{enumerate}[label=(\roman*), labelwidth=\widthof{(iii)}]
  \item
    If $u=1$, then $[u^2,1]=\{1\}$, so the unique solution is $x_1^\ast=1$.
    Fix now $u\in(0,1)$. Equation~\eqref{eq:singular:curve} is equivalent to
    $h_u(x)=0$, where
    \begin{align*}
      h_u(x)=\alpha\ln(1-x)-\beta\ln\!\left(1-\frac{u^2}{x}\right),
      \qquad x\in(u^2,1).
    \end{align*}
    Since
    \begin{align*}
      h_u'(x)
      =
      -\frac{\alpha}{1-x}
      -\frac{\beta u^2}{x(x-u^2)}
      <0,
      \qquad x\in(u^2,1),
    \end{align*}
    the function $h_u$ is strictly decreasing. Moreover, $\lim_{x\downarrow u^2} h_u(x)=+\infty$ and $
      \lim_{x\uparrow 1} h_u(x)=-\infty$.
    Therefore, by the intermediate value theorem, there exists a unique
    $x_u^\ast\in(u^2,1)$ satisfying $h_u(x_u^\ast)=0$.

  \item
    We first show that $x_u^\ast\to0$ and $u^2/x_u^\ast\to0$ as $u\downarrow0$.
    Suppose that $\limsup_{u\downarrow0}x_u^\ast>0$. Then there exist
    $\varepsilon>0$ and a sequence $u_n\downarrow0$ such that
    $x_{u_n}^\ast\ge\varepsilon$. Hence
$(1-x_{u_n}^\ast)^\alpha\le (1-\varepsilon)^\alpha<1$,
    while
$      \left(1-u_n^2/x_{u_n}^\ast\right)^\beta\to1$,
contradicting~\eqref{eq:singular:curve}. Since $x_u^\ast\ge u^2>0$, this proves
    $x_u^\ast\to0$.

    Similarly, suppose that $\limsup_{u\downarrow0}(u^2/x_u^\ast)>0$. Then there exist
    $\varepsilon>0$ and a sequence $u_n\downarrow0$ such that
    $u_n^2/x_{u_n}^\ast\ge\varepsilon$. Then $\left(1-u_n^2/x_{u_n}^\ast\right)^\beta\le (1-\varepsilon)^\beta<1$,
    while
  $ (1-x_{u_n}^\ast)^\alpha\to1,
    $
    again contradicting~\eqref{eq:singular:curve}. Hence $u^2/x_u^\ast\to0$.

    Next define, for $p>0$,
    \begin{align*}
      g_p(z):=\frac{1-(1-z)^p}{z},\qquad z\in(0,1).
    \end{align*}
    Then $g_p(z)\to p$ as $z\downarrow0$. Since $x_u^\ast$ satisfies
    \eqref{eq:singular:curve}, we have
    $
      1-(1-x_u^\ast)^\alpha
      =
      1-\left(1-u^2/x_u^\ast\right)^\beta,
    $
    and therefore
    \begin{align*}
      g_\alpha(x_u^\ast)\,x_u^\ast
      =
      g_\beta\!\left(\frac{u^2}{x_u^\ast}\right)\frac{u^2}{x_u^\ast}.
    \end{align*}
    By rearranging this equation, we obtain
    \begin{align*}
      \left(\frac{x_u^\ast}{u}\right)^2
      =
      \frac{
        g_\beta\!\left(u^2/x_u^\ast\right)
      }{
        g_\alpha(x_u^\ast)
      }.
    \end{align*}
    Since $x_u^\ast\to0$ and $u^2/x_u^\ast\to0$, we conclude that
    $\lim_{u\downarrow0}\left(x_u^\ast/u\right)^2=\beta/\alpha$, that is
    $x_u^\ast \simeq u\sqrt{\beta/\alpha}$.

    Finally, by Theorem~\ref{thm:equivalence}~\ref{item:varphi:ast:asymptotic}
    and Lemma~\ref{lem:tail:copula:smo},
    any function of maximal dependence $\varphi^\ast$ satisfies
    $\varphi^\ast(u)/u\to\sqrt{\beta/\alpha}$ as $u\downarrow0$.
    Hence $\varphi^\ast(u)\simeq x_u^\ast \simeq u\sqrt{\beta/\alpha}$ as $u\downarrow0$. \qedhere
  \end{enumerate}
\end{proof}

\end{document}